\documentclass[11pt,a4paper]{article}
\usepackage[english]{babel} 
\usepackage{amsmath}
\usepackage{amsfonts}
\usepackage{amssymb}
\usepackage{amsthm}
\usepackage{makeidx}
\usepackage{graphicx} 
\usepackage{natbib}
\usepackage{setspace}
\usepackage{float}
\usepackage{hyperref}
\usepackage[margin=1in]{geometry}
\usepackage{lscape} 
\usepackage{setspace}
\usepackage[colorinlistoftodos]{todonotes} 
\hypersetup{
	colorlinks   = true,
	citecolor    = blue,
	urlcolor = blue
}

\newcommand{\R}{\mathbb R} 
 
\usepackage{tabularx}

\usepackage{ragged2e}
\title{On a robust risk measurement approach for capital determination errors minimization}

\author{Marcelo Brutti Righi$^{a}$\footnote{Corresponding author. We would like to thank the editor, anonymous associate editor, and reviewer for constructive comments and suggestions, which
		have been very useful to improve the technical quality of the manuscript. We are grateful for the financial support of FAPERGS (Rio Grande do Sul State
		Research Council) project number 17/2551-0000862-6 and CNPq (Brazilian Research Council) projects number
		302369/2018-0 and 407556/2018-4.}\\  \small{ \href{mailto:marcelo.righi@ufrgs.br}{marcelo.righi@ufrgs.br}} \and  Fernanda Maria M\"{u}ller$^{a}$\\\small{\href{mailto:fernanda.muller@ufrgs.br}{fernanda.muller@ufrgs.br}}
\and Marlon Ruoso Moresco$^{a}$\\\small{\href{mailto:marlon.moresco@ufrgs.br}{marlon.moresco@ufrgs.br}}}

\date{\small{$^{a}$\textit{Business School, Federal University of Rio Grande do Sul, Washington Luiz, 855, Porto Alegre, Brazil, zip 90010-460}}}

\newtheorem{Def}{Definition}
\newtheorem{Thm}{Theorem}
\newtheorem{Prp}{Proposition}
\theoremstyle{definition}
\newtheorem{Exm}{Example}

\theoremstyle{remark}
\newtheorem{Rmk}{Remark}

\begin{document}
	
	\maketitle
	
	\begin{abstract}
 We propose a robust risk measurement approach that minimizes the expectation of overestimation plus underestimation costs. We consider uncertainty by taking the supremum over a collection of probability measures, relating our approach to dual sets in the representation of coherent risk measures.
We provide results that guarantee the existence of a solution and explore the properties of minimizer and minimum as risk and deviation measures, respectively. An empirical illustration is carried out to demonstrate the use of our approach  in capital determination. 
	\end{abstract}
	
\section{Introduction}\label{introduction}
	
The interest in risk measures with interpretation as capital determination from a theoretical point of view has been raised in mathematical finance and insurance since the seminal paper of \cite{Artzner1999}. 
 From there, an entire stream of literature has proposed and discussed distinct features, including  axiom sets, dual representations, 
mathematical and statistical properties.
We suggest  \cite{follmer2015axiomatic} and \cite{Follmer2016} for a recent review of this theory. 
 See \cite{Emmeretal:2015} for a discussion and comparison of    properties of popular risk measures, which include  variance, Value at Risk (VaR), Expected Shortfall (ES), and Expectile Value at Risk (EVaR).

 Despite the investigations carried out in this regard, there is still no consensus about a definitive set of properties or the best risk measure for practical matters.
 In this context, there is scope for proposing  new risk measurement approaches, such as \cite{Righi2016}, \cite{furman2017gini},  \cite{Righi2018a}, \cite{Bellini2019} and \cite{Pichler2019}. A possible interesting risk measurement process could seek to minimize capital determination errors to reduce the costs linked to it. From the regulatory point of view, risk underestimation, and consequently capital determination underestimation, is the main concern. In this case, capital charges are desirable  to avoid costs from unexpected and uncovered losses.
However, from the perspective of institutions, it is also desirable to reduce the regret costs arising from risk overestimation because the latter reduces profitability.

Based on this perspective, we propose a risk measurement procedure that represents the capital determination for a financial position  $X$ that minimizes the expected value of the sum between costs from risk overestimation (regret for lesser profits) and underestimation (uncovered losses). We measure these two costs in our framework by positive random variables  $G$, for gains, and $L$, for losses. These costs may refer, for instance, to financial rates traded in the market or even as a function of variables of interest. Typically, such costs can be interpreted as the cost of opportunity measured by the rate of a similar investment or even risk free rate of the exceeding capital, in the case of $G$, and the cost of raising money, by borrowing from bank or other debt sources, in the case of underestimation to cover an unexpected loss, in the case of $L$. More specifically, our goal is to minimize for $x\in\mathbb{R}$ the expectation of $(X-x)^{+}G+(X-x)^{-}L$. The term $(X(\omega)-x)^{+}G(\omega)$ represents a situation where the realized outcome of $X$ is better than the capital requirement and the surplus could be invested at cost $G$. Similarly, $(X(\omega)-x)^{-}L(\omega)$ relates to a situation where the capital reserve is not enough to cover the loss, where the difference must be raised at cost $L$.  

Thus, in this paper we have four contributions to both academic literature and the financial industry:
\begin{enumerate}
	\item The type of loss function we consider has not been considered for such purpose. In this sense, \cite{Laeven2004} and \cite{Goovaerts2005} explore a loss function more similar to ${-xG+(X-x)^{-}}$. 	\cite{dhaene2003economic},   \cite{Goovaerts2005}, and  \cite{Goovaerts2010} present, even briefly, a loss function similar to ${-xG+(X-x)^{-}L}$. It is worth to mention also the works of \cite{Zaks2006}, \cite{Dhaene2012}, \cite{xu2012stochastic}, and \cite{xu2013optimal}, among other research on this topic. Such studies also focus on minimizing the mentioned costs in order to obtain the optimal amount of capital. The term $-xG$ is the cost of putting aside the
	initial capital $x$ instead of investing it in a financial instrument
	with returns payoff $G$, which is not a random variable in these studies. This is distinct to regret or overestimation costs as in our framework since $(X-x)^+G$ has not its realization known a priori. In our case, we penalize regret costs only the exceeding value over capital determination. Thus, our focus is more on risk measurement errors. If the determined capital (risk measure) is precisely the monetary value that we need to cover the loss, then there is no reason for penalization. On the other hand, most of these works focus on determining the optimal composition for aggregated capital of a financial firm from its business units. 	We emphasize that such approaches are originally proposed to determine economic capital for insurance analysis. 
	
	\item We provide results that guarantee a solution for the optimization problem linked to our risk measure, develop properties that it fulfills, and characterize the resulting minimum cost as a deviation measure in the sense of \cite{rockafellar2006generalized}. In this sense, our approach considers a risk measure that has random variables, the costs $G$ and $L$, as parameters. With very few exceptions, such as the maximal correlation risk measure in \cite{Ruschendorf2006}, parameters for risk measures are real numbers. This is the case of VaR and ES, for instance, where the quantile significance level $\alpha\in(0,1)$ is the parameter. In fact, when $G=\alpha$ and $L=1-\alpha$, our risk measure coincides with VaR, and its related minimum is a scaled ES deviation from the expectation. This kind of generalization leads to technical difficulties we handle, such as dependence between the financial position $X$ and costs $G$ and $L$. The paper of \cite{Rockafellar2013} relates to risk and deviation measures linked by a common optimization problem but does not consider random variables as parameters. \cite{Bellini2014a} study generalized quantiles as risk measures by minimizing asymmetric loss functions but also do not consider random variables as parameters.
	
	\item Our approach is robust in the sense we consider a supremum of probability measures over our optimization problem. This is a
worst case approach, where the functionals are not sensitive to the choice of some specific probability measure that represents a particular belief about the world. In this sense, we are in concordance to the stream of \cite{Shapiro2017}, \cite{Righi2018b}, \cite{Bellini2018},  and \cite{Guo2018}. Nonetheless, none of them consider the same features we address in our study, such as random variables as parameters and the minimum as a deviation. Moreover, we relate our approach with the dual representation of coherent risk measures, since we are based on a supremum of expectations over probabilities. For model risk purposes, \cite{cont2006model} uses this dual representation scheme for robust expectations regarding pricing contingent claims.
	
\item We provide an illustration considering real financial data with the purpose of showing the practical usefulness of our approach. In this sense, we perform an adaptation of the dual representation of usual discrete probability spaces. We then consider capital requirements determined by our risk measures against usual risk measures applied for this purpose. Results allow us to conclude that capital determination based on typical tail risk measures that the Basel and Solvency accords recommend and by risk measures that also focus on minimizing capital costs that are too punitive, lead to more costs concerning risk measures obtained from our approach.
\end{enumerate}	
	
Regarding structure, the remainder of this paper divides in the following contents: in Section \ref{sec:theory} we expose definitions and results regarding the existence of solutions to the optimization problems in our proposed approach; in Section \ref{sec:prop} we prove results regarding the properties of our risk and deviation measures in relation to financial positions and costs for capital determination errors; in Section \ref{sec:illust} we exhibit an  empirical illustration of our approach for capital determination; 
in Section \ref{conclusion} we summarize and conclude the paper.

\section{Proposed approach}\label{sec:theory}

The content is based on the following notations. Consider the real valued random result $X$ of any asset ($X\geq0$ is a gain, $X<0$ is a loss) that is defined on a probability space $(\Omega,\mathcal{F},\mathbb{P})$. All equalities and inequalities are considered almost surely in $\mathbb{P}$. We define $X^+=\max(X,0)$, $X^-=\max(-X,0)$, and $1_A$ as the indicator function for an event  $A$.

We let $\mathcal{Q}$ denote the set composed of probability measures  $\mathbb{Q}$ defined on $(\Omega,\mathcal{F})$ that are absolutely continuous with respect to $\mathbb{P}$, with Radon-Nikodym derivative $d\mathbb{Q}/d\mathbb{P}$ and $\mathcal{Q}^\prime\subseteq\mathcal{Q}$ a non-empty set.  Moreover, $E_{\mathbb{Q}}[X]=\int_{\Omega}Xd\mathbb{Q}$, $F_{X, \mathbb{Q}}(x)=\mathbb{Q}(X\leq x)$ and $F_{X, \mathbb{Q}}^{-1}(\alpha)=\inf\{x\in\mathbb{R}\colon F_{X, \mathbb{Q}}(x)\geq\alpha\}$ are, respectively, the expected value, the distribution function and the (left) quantile of $X$ under $\mathbb{Q}$. We drop the subscript  when it is regarding  $\mathbb{P}$. 

Let $L^{\infty}(\mathbb{Q}):=L^{\infty}(\Omega,\mathcal{F},\mathbb{Q})$ be the vector space of (equivalent classes of) $\mathbb{Q}-a.s.$ bounded random variables. We write $L^\infty:=L^\infty(\mathbb{P})$. We have that $L^{\infty}_{+}$ and $L^\infty_{++}$ are its cones of non-negative and positive elements, respectively. We denote by $X_n\rightarrow X$ convergence in the $L^\infty$ essential supremum norm, while $\lim\limits_{n\rightarrow\infty}X_n=X$ means $\mathbb{P}$-a.s. convergence.

We recurrently use the following facts for any $\mathbb{Q}\in\mathcal{Q}$ without further mention:
\begin{itemize}
	\item $L^\infty(\Omega,\mathcal{F},\mathbb{P})\subseteq L^\infty (\Omega,\mathcal{F},\mathbb{Q})$.
	\item $\operatorname{ess}\sup_\mathbb{P} X\geq\operatorname{ess}\sup_\mathbb{Q}X$ and $\operatorname{ess}\inf_\mathbb{P} X\leq\operatorname{ess}\inf_\mathbb{Q}X$, $\forall\:X\in L^\infty$. 
	\item if $F_{X}$ is continuous, then also is $F_{X,\mathbb{Q}},\:\forall\:X\in L^\infty$.
	\item  if $X_n\rightarrow X$ or $\lim\limits_{n\rightarrow\infty}X_n=X$ regarding $\mathbb{P}$, then the same is true concerning $\mathbb{Q}$ for any $\{X_n\}\subset L^\infty$ and $X\in L^\infty$.
\end{itemize}

We now define our risk measurement approach as the supremum of minimization problems regarding costs from capital determination errors.

\begin{Def}\label{def:meas} Let $G,L\in L^{\infty}_{++}$. Then:
\begin{enumerate}
\item 	Our risk measure is a functional $R\colon L^{\infty} \rightarrow\mathbb{R}$ defined as:
\begin{equation}
R(X):=R_{G,L,\mathcal{Q}^\prime}(X)=\sup\limits_{\mathbb{Q}\in\mathcal{Q}^\prime}\left\lbrace -\min\left\lbrace\arg\min\limits_{x\in\mathbb{R}} E_{\mathbb{Q}}[(X-x)^{+}G+(X-x)^{-}L]\right\rbrace\right\rbrace \label{averserisk}.
\end{equation}
\item  
	Our deviation measure is a functional $RD\colon L^{\infty} \rightarrow\mathbb{R}_+$ defined as:	
\begin{equation}
RD(X):=RD_{G,L,\mathcal{Q}^\prime}(X)=\sup\limits_{\mathbb{Q}\in\mathcal{Q}^\prime}\left\lbrace \min\limits_{x\in\mathbb{R}} E_{\mathbb{Q}}[(X-x)^{+}G+(X-x)^{-}L]\right\rbrace\label{aversedev}.
\end{equation}
\end{enumerate}
\end{Def}

The negative sign for $R$ is to keep the pattern for losses. Moreover, in Propositions \ref{prp:sol2} and \ref{prp:sol} we prove that our both functionals are in fact finite and, thus, well-defined. It would also be interesting to extend our results to general loss functions $f_1((X-x)^+)G+f_2((X-x)^-)L$, with $f_1$ and $f_2$ beyond the identity function as in \cite{Bellini2014a}. We do not pursue this goal because we want to keep the intuitive meaning of our specific loss function.

\begin{Rmk}
A recently highlighted statistical property is Elicitability, which enables the comparison of competing models in risk forecasting. See \cite{ziegel2016coherence} and the references therein for more details. A functional is elicitable if it is the argmin for expectation of some score function $S\colon \mathbb{R}^2\rightarrow\mathbb{R}_+$ and satisfies certain properties.
For instance, mean and $\alpha$-quantile are elicitable under scores $(x-y)^2$ and $\alpha(x - y)^+ + (1 - \alpha)(x - y)^-$, respectively. However our functionals are not fitted in this class because costs $G$ and $L$ are random variables that are not necessarily functions of the position $X$. A possibility could be to consider a score function defined as $S(x,y)=(x-y)^+g(x)+(x-y)^-l(x)$ with $g,l\colon\mathbb{R}\rightarrow\mathbb{R}$ holding certain properties. However, we do not pursue such framework since our focus is precisely into considering $G$ and $L$ as general random variables that are not strictly dependent to $X$. 
\end{Rmk}

\begin{Rmk}\label{remarkfinite}
It is straightforward that $(X-x)^{+}G+(X-x)^{-}L\in L^\infty$ for any $X\in L^\infty$ and any $x\in\mathbb{R}$.  It would be possible to develop our theory by considering the domain of both $R$ and $RD$ as the $L^1$ vector space of $\mathbb{P}$-integrable random variables since this is a larger space. In such case, we would have to restrict ourselves to probability measures $\mathcal{Q}^\prime$ with $\mathbb{P}$-a.s. bounded Radon-Nikodym derivatives, i.e., $d\mathbb{Q}/d\mathbb{P}\in L^\infty$. Such technicality is in order to guarantee integrability since for any $\mathbb{Q}$ under these conditions we have from H{\"o}lder inequality that \[E_\mathbb{Q}[(X-x)^{+}G+(X-x)^{-}L]\leq \left( E[(X-x)^{+}]\operatorname{ess}\sup G+ E[(X-x)^{-}]\operatorname{ess}\sup L\right) \operatorname{ess}\sup d\mathbb{Q}/d\mathbb{P},\] which is finite. Most results we expose are easily adaptable to such framework.  
\end{Rmk}

We now expose a formal result that guarantees our minimization problems have a solution. Note that $(X-x)^{+}\geq(X-\operatorname{ess}\inf X)^{+}$ for $x<\operatorname{ess}\inf X$, and $(X-x)^{-}\geq(X-\operatorname{ess}\sup X)^{-}$ for $x>\operatorname{ess}\sup X$. Thus, the minimization problem would not be altered when we take optimization over the compact interval $[\operatorname{ess}\inf X, \operatorname{ess}\sup X]$ instead of the whole real line because we consider random variables in $L^{\infty}$.

\begin{Prp}\label{prp:prob}
Under notation of Definition \ref{def:meas}, let $B_\mathbb{Q}:=\arg\min\limits_{x\in\mathbb{R}} E_{\mathbb{Q}}[(X-x)^{+}G+(X-x)^{-}L]$. Then for each $\mathbb{Q}\in\mathcal{Q}$:
\begin{enumerate}
	\item $B_\mathbb{Q}$ is a closed interval.
	\item $x\in B_\mathbb{Q}$ if and only if $x$ satisfies the first order condition given by 
	\[\begin{cases}
E_\mathbb{Q}[G1_{\{X\geq x\}}]\geq E_\mathbb{Q}[L1_{\{X< x\}}]\\
E_\mathbb{Q}[G1_{\{X>x\}}]\leq E_\mathbb{Q}[L1_{\{X\leq x\}}].
	\end{cases}\]
	\item if $F_{X,\mathbb{Q}}$ is  continuous in $B_\mathbb{Q}$, then $B_\mathbb{Q}$ is a singleton if and only if  $F_{X,\mathbb{Q}}$ is strictly increasing in $B_\mathbb{Q}$.
\end{enumerate} 
\end{Prp}

\begin{proof}
 Fix $\mathbb{Q}\in\mathcal{Q}$. Then:
\begin{enumerate}

\item Let $f_{X,\mathbb{Q}}:  \R \rightarrow \R$ be defined as  $f_{X,\mathbb{Q}} (x):= E_{\mathbb{Q}}[(X-x)^{+}G+(X-x)^{-}L] $. Clearly, $f_{X,\mathbb{Q}}$ is finite, due to  Remark  \ref{remarkfinite}, convex and, hence a continuous function.  Note that  $f_{X,\mathbb{Q}}$ is  proper and level bounded. Thus,  $\inf_{x\in\mathbb{R}}  f_{X,\mathbb{Q}}(x)$ is finite and the set $\arg\min_{x\in\mathbb{R}}  f_{X,\mathbb{Q}}(x)$ is non-empty and compact. Moreover, since $f_{X,\mathbb{Q}}$ is convex, $B_\mathbb{Q}$ is an interval.

\item We have to solve the first order condition for $x$ in order to obtain the argument that minimizes the expression. Note that $ f_{X,\mathbb{Q}}$ is convex. Then $y\in\mathbb{R}$ is a minimizer if and only if
 \[0\in\left[\dfrac{\partial^-  f_{X,\mathbb{Q}}}{\partial x}(y),\dfrac{\partial^+  f_{X,\mathbb{Q}}}{\partial x}(y)\right].\] 
 Thus, we have from Dominated convergence that both \begin{align*}
\dfrac{\partial^- f_{X,\mathbb{Q}}}{\partial x}&=E_{\mathbb{Q}}\left[ \dfrac{\partial^- [(X-x)^{+}G+(X-x)^{-}L]}{\partial x}\right]\\
&= E_{\mathbb{Q}}[-G1_{\{X\geq x\}}+L1_{\{X<x\}}],
\end{align*}  and \begin{align*}
\dfrac{\partial^+ f_{X,\mathbb{Q}}}{\partial x}&=E_{\mathbb{Q}}\left[ \dfrac{\partial^+ [(X-x)^{+}G+(X-x)^{-}L]}{\partial x}\right]\\
&= E_{\mathbb{Q}}[-G1_{\{X>x\}}+L1_{\{X\leq x\}}].
\end{align*} By taking $\dfrac{\partial^-  f_{X,\mathbb{Q}}}{\partial x}\leq0$ and $\dfrac{\partial^+  f_{X,\mathbb{Q}}}{\partial x}\geq0$, we get the claim.

\item Note that when $F_{X,\mathbb{Q}}$ is continuous in $x\in\mathbb{R}$ we have $\mathbb{Q}(X=x)=0$. Then, the first order condition of (ii) is equivalent to $E_\mathbb{Q}[G1_{\{X\geq x\}}]= E_\mathbb{Q}[L1_{\{X\leq x\}}]$, which can be rewritten as $E_\mathbb{Q}[G]=E_\mathbb{Q}[(G+L)1_{\{X\leq x\}}]$. For the if part, 
assume that there exist $x,y\in B_\mathbb{Q}$ and $x<y$, i.e., $ B_\mathbb{Q}$ is not a singleton. Since  $F_{X,\mathbb{Q}}$ is strictly increasing in  $B_\mathbb{Q}$, we get that $\mathbb{Q}(y\geq X>x)>0$ and, thus,  $\mathbb{Q}(1_{\{y\geq X>x\}}>0)>0$. Now, from $G,L >0 $ we have $ E_\mathbb{Q}[(G+L)1_{\{y\geq X> x\}}] >0$. Therefore, 
\[
E_\mathbb{Q}[G]=E_\mathbb{Q}[(G+L)1_{\{X\leq x\}}]= E_\mathbb{Q}[(G+L)(1_{\{X\leq y\}}-1_{\{y\geq X> x\}})] <  E_\mathbb{Q}[(G+L)1_{\{X\leq y\}}].\]
Thus, $y $ does not  fulfill the first order condition, a contradiction. The reasoning is analogous for $y<x$. Hence, $B_\mathbb{Q}$ is a singleton, as desired. For the only if part,
let $x \in B_\mathbb{Q}$ and $F_{X,\mathbb{Q}}$ be continuous in $B_\mathbb{Q}$ but constant in some half ball $[x,x+r)$ with radius $r$ around $x$, i.e., $F_{X,\mathbb{Q}}$ is not strictly increasing in $B_\mathbb{Q}$.  Then, for any $x<y<x+r$, we have that $1_{\{X\leq x\}} = 1_{\{X\leq y\}}$. Thus, it is direct that
\[
E_\mathbb{Q}[G]=E_\mathbb{Q}[(G+L)1_{\{X\leq x\}}]=E_\mathbb{Q}[(G+L)1_{\{X\leq y\}}].
\]
Therefore, $y$ also fulfills the first order condition and $B_\mathbb{Q}$ is not a singleton. 

\end{enumerate}
\end{proof}

\begin{Rmk}
Equivalent forms for the first order condition of item (ii) are available, such as $E_\mathbb{Q}[(G+L)1_{\{X\leq x\}}]\geq E_\mathbb{Q}[G]\geq E_\mathbb{Q}[(G+L)1_{\{X< x\}}]$. We have in fact used it in the proof of item (iii) and will keep using it without further mention. Moreover, continuity of $F_{X,\mathbb{Q}}$ is not a necessary condition for the argmin set to be a singleton. For instance, Take $G := L =
0.5$ and $X$ with $\mathbb{Q}(X = k) = 0.5, \mathbb{Q}(X > k) = \mathbb{Q}(X < k) = 0.25$ for some $k\in\mathbb{R}$. Then $k=F^{-1}_{X,\mathbb{Q}}(0.5)$ is the unique minimizer despite that $F_{X,\mathbb{Q}}$ is not continuous at $k$.
\end{Rmk}

We now derive a solution in terms of quantile functions under the assumption that position $X$ is independent of both costs $G$ and $L$ regarding any probability measure in $\mathcal{Q}^\prime$. Note that this result rules out convexity from $R$ since VaR is not a convex risk measure.
\begin{Prp}\label{prp:rob}
	Under the notation in Definition \ref{def:meas}, let $X\in L^{\infty}$ be independent of both $G$ and $L$ regarding any probability measure in $\mathcal{Q}^\prime$ and $\alpha(\mathbb{Q})=\dfrac{E_\mathbb{Q}[G]}{E_\mathbb{Q}[G+L]}$. Then:
	\begin{enumerate}
		\item \begin{equation}\label{eq:sol}
		R(X)=\sup_{\mathbb{Q}\in\mathcal{Q}^\prime}\left\lbrace  -F_{X,\mathbb{Q}}^{-1} \left(\alpha(\mathbb{Q})\right)\right\rbrace .
		\end{equation} 
		\item \begin{equation}\label{eq:sol2}
		RD(X)=\sup_{\mathbb{Q}\in\mathcal{Q}^\prime}\left\lbrace E_\mathbb{Q}[G]\left( E_\mathbb{Q}[X]-\frac{1}{\alpha(\mathbb{Q})}\int_{0}^{\alpha(\mathbb{Q})}F^{-1}_{X,\mathbb{Q}}(s)ds\right) \right\rbrace .
		\end{equation} 
	\end{enumerate}

\end{Prp}

\begin{proof}

	\begin{enumerate}
		\item We have that the first order condition in Proposition \ref{prp:prob} becomes for each $\mathbb{Q}\in\mathcal{Q}^\prime$	\[\mathbb{Q}(X< x)\leq\alpha(\mathbb{Q})\leq\mathbb{Q}(X\leq x),\] which is valid since $E_\mathbb{Q}[G+L]>0$. Thus, for $x\in B_\mathbb{Q}$, $x$ is a $\alpha(\mathbb{Q})$ - quantile of $X$. We then get that the minimum value $x\in\mathbb{R}$ that satisfies such condition is \[\min\left\lbrace x\in\mathbb{R}\colon F_{X,\mathbb{Q}}
		(x)\geq\alpha(\mathbb{Q}) \right\rbrace=F_{X,\mathbb{Q}}^{-1} \left(\alpha(\mathbb{Q}) \right).\] Applying the supremum over $\mathcal{Q}^\prime$ leads to the desired result.
		\item By letting $R=F_{X,\mathbb{Q}}^{-1}\left(\alpha(\mathbb{Q})\right)$ we see that the minimum becomes, for each $\mathbb{Q}\in\mathcal{Q}^\prime$:
		\begin{align*}
		&E_\mathbb{Q}\left[\left(X-R\right)^{+}G+\left(X-R\right)^{-}L\right]\\
		=&E_\mathbb{Q}\left[\left(X-R\right)G+\left(X-R\right)^{-}G+\left(X-R\right)^{-}L\right]\\
		=&E_\mathbb{Q}[G]\left(E_\mathbb{Q}[X]-R+E_\mathbb{Q}\left[\left(R-X\right)1_{\{X\leq R\}}\right]\frac{1}{\alpha(\mathbb{Q})}\right)\\
		=&E_\mathbb{Q}[G]\left(E_\mathbb{Q}[X]-\frac{1}{\alpha(\mathbb{Q})}\left( \alpha(\mathbb{Q}) R+E_\mathbb{Q}\left[\left(X-R\right)1_{\{X\leq R\}}\right]\right) \right)\\
		=&E_\mathbb{Q}[G]\left(E_\mathbb{Q}[X]-\left( R+\frac{1}{\alpha(\mathbb{Q})}E_\mathbb{Q}\left[-\left(R-X\right)^+\right] \right)\right)\\
		=&E_\mathbb{Q}[G]\left( E_\mathbb{Q}[X]-\frac{1}{\alpha(\mathbb{Q})}\int_{0}^{\alpha(\mathbb{Q})}F^{-1}_{X,\mathbb{Q}}(s)ds\right).
		\end{align*}
		The equivalence between expectation and integral in the last step is due to the optimization formula for Expected Shortfall (also known as Conditional Value at Risk or Average Value at Risk),  see Proposition 4.51 of \cite{Follmer2016} for instance. By applying the supremum over $\mathcal{Q}^\prime$ we get the desired claim. 
	\end{enumerate}

\end{proof}

\begin{Rmk}
	When $G=\beta$ and $L=1-\beta$, we obtain as solution the $\beta$-quantile (VaR). We recommend \cite{Bellini2014a} for details. 
	 When $F_{X,\mathbb{Q}}$ is continuous at $\alpha(\mathbb{Q})$ for any $\mathbb{Q}\in\mathcal{Q}^\prime$, we get \[RD(X)=\sup_{\mathbb{Q}\in\mathcal{Q}^\prime}\left\lbrace -E_{\mathbb{Q}}[G]E_{\mathbb{Q}}\left[X-E_{\mathbb{Q}}[X]\bigg|X\leq F_{X,\mathbb{Q}}^{-1} \left( \alpha(\mathbb{Q})\right)\right]\right\rbrace,\] which is directly linked to the tail mean (ES) deviation by representing the distance between regular and tail expectations. Moreover, when $\mathcal{Q}^\prime=\mathcal{Q}$ we may end up to be in a situation where $G$ and $L$ are constants. Nonetheless, for the general situation where $\mathcal{Q}^\prime$ can be any strict subset of $\mathcal{Q}$ (a particular example is a singleton $\mathcal{Q}^\prime=\{\mathbb{Q}\}$) we are not necessarily limited to $G$ and $L$ constants. For instance, let $A,B\in\mathcal{F}$ be such that $\mathbb{Q}(A\cap B)=\mathbb{Q}(A)\mathbb{Q}(B)$ and $0<\mathbb{Q}(A),\mathbb{Q}(B)<1$ for any $\mathbb{Q}\in\mathcal{Q}^\prime$. Also, let $X:=1_A$ and $G:=L:=1_B+k>0$ for some $k\in(0,\infty)$. Then, regarding any $\mathbb{Q}\in\mathcal{Q}^\prime$, $X$ is independent of both $G$ and $L$, while $G$ and $L$ are not constants. 
\end{Rmk}

Regarding the ambiguity set, one can choose $\mathcal{Q}^\prime\subseteq\mathcal{Q}$ 
in an ad hoc sense according to some \textit{a priori} established risk aversion parameter. Another possibility is to consider those probability measures that represent beliefs absolutely continuous inside some distance from a nominal measure $\mathbb{P}$, as in \cite{Shapiro2017}. It is straightforward to note that $\mathcal{Q}^\prime_1\subseteq\mathcal{Q}^\prime_2\subseteq\mathcal{Q}$ implies both $R_{\mathcal{Q}^\prime_1}\leq R_{\mathcal{Q}^\prime_2}$ and $RD_{\mathcal{Q}^\prime_1}\leq RD_{\mathcal{Q}^\prime_2}$. 

We now expose a particular choice for $\mathcal{Q}^\prime$ linked to dual representations of coherent risk measures (sub-linear expectations as in \cite{Sun2017}). These kind of functionals possess the properties of Monotonicity, Translation Invariance, Positive Homogeneity and Convexity.  We now present a formal result that guarantees the dual representation of  coherent risk measures in $L^{\infty}$. 

\begin{Thm}[\cite{Artzner1999}, \cite{Delbaen2002}]	A functional $\rho : L^{\infty}\rightarrow\mathbb{R}$ is a lower semi-continuous in $\mathbb{P}-a.s.$ sense for bounded sequences coherent risk measure if and only if it can be represented as
	\begin{equation*}
	\rho(Y)=\sup\limits_{\mathbb{Q}\in\mathcal{Q}_{\rho}} E_{\mathbb{Q}} [-Y],\:\forall\;Y\in L^{\infty},
	\end{equation*}
	where  $\mathcal{Q}_{\rho}\subseteq\mathcal{Q}$ is a closed and convex non-empty set, called dual set of $\rho$. If $\Omega$ is finite, then the same is true without lower semi-continuity in $\mathbb{P}-a.s.$ sense.
\end{Thm}


\begin{Exm}\label{ex:risk}
	Possible, but not limited, choices of $\rho$ are:
	\begin{itemize}
		\item Expected Loss (EL): This risk measure is defined as $EL(Y)=E[-Y]$. It is the most parsimonious one, indicating the expected value (mean) of a loss. Its dual set is a singleton $\mathcal{Q}_\rho=\{\mathbb{P}\}$, i.e., only consider the basic belief. 
		
		\item Mean plus Semi-Deviation (MSD): This risk measure is defined as $MSD^{\beta}(Y)=-E[Y] + \beta \sqrt{E[((Y-E[Y])^{-})^2]}, 0\leq\beta\leq1$. The advantages of this risk measure are its simplicity and financial meaning.  $\mathcal{Q}_{MSD^\beta}=\left \lbrace\mathbb{Q} \in \mathcal{Q} : \frac{d\mathbb{Q}}{d\mathbb{P}} =1+ \beta( V-E[V]), V\geq0, E[|V|^2]=1 \right\rbrace$ represents its dual set. 

		\item Expected Shortfall (ES): This risk measure is defined as $ES^{\alpha}(Y)=-\frac{1}{\alpha}\int_{0}^\alpha F^{-1}_X(s)ds, 0<\alpha<1$. It represents the expected value of a loss, given it is beyond the  $\alpha$ - quantile of interest. Its dual set is $\mathcal{Q}_{ES^\alpha}=\left\lbrace\mathbb{Q}\in\mathcal{Q} : \frac{d\mathbb{Q}}{d\mathbb{P}}\leq\frac{1}{\alpha} \right\rbrace$. ES is the most utilized coherent risk measure, being the basis of many representation theorems in this field.
		
		\item Expectile Value at Risk (EVaR): This measure links to the concept of an expectile, given by  $EVaR^{\alpha}(Y)=-\arg\min\limits_{x\in\mathbb{R}} E[\alpha((X-x)^+)^2+(1-\alpha)((X-x)^-)^2],\:0<\alpha\leq0.5$. Its dual set is $\mathcal{Q}_{EVaR^\alpha}=\left\lbrace\mathbb{Q}\in\mathcal{Q} \colon\:\exists\: a>0,\: a\leq\frac{d\mathbb{Q}}{d\mathbb{P}}\leq a\frac{1-\alpha}{\alpha} \right\rbrace$. EVaR is the only coherent risk measure, beyond EL, that possesses the property of elicitability. 
		
		\item Maximum loss (ML): This is an extreme risk measure that has dual set $\mathcal{Q}_{ML}=\mathcal{Q}$, i.e., all the considered beliefs. We define it as  $ML(Y)=-\operatorname{ess}\inf Y $. Such a risk measure leads to the more protective situations since  $ML(Y)\geq\rho(Y)$ for any coherent risk measure $\rho$.
	\end{itemize} 
\end{Exm}

Formulations \eqref{averserisk} and \eqref{aversedev} can be considered in situations with   $\mathcal{Q}^\prime=\mathcal{Q}_{\rho}$, where $\rho$ is a coherent risk measure. We define both $R_{\mathcal{Q}_\rho}:=R_{\rho}$ and $RD_{\mathcal{Q}_\rho}:=RD_{\rho}$ in order to ease notation. The next Proposition presents an alternative formulation for $RD$.

\begin{Prp}\label{prp:risk}
	Let $\mathcal{Q}^\prime=\mathcal{Q}_\rho$, where $\rho:L^\infty\rightarrow\mathbb{R}$ is a lower semi-continuous in $\mathbb{P}-a.s.$ sense for bounded sequences coherent risk measure. Then:
	\begin{equation}
		RD_{\rho}(X)=\min\limits_{x\in\mathbb{R}}\rho(-(X-x)^{+}G-(X-x)^{-}L),\:\forall\:X\in L^\infty\label{eq:rho2}.
		\end{equation}
\end{Prp}

\begin{proof}
 The result follows from the Sion's minimax theorem, see \cite{Sion1958}, because the map $(x,\mathbb{Q})\rightarrow E_\mathbb{Q}[(X-x)^+G+(X-x)^-L]$ has the needed continuity and quasi-convex properties, $\mathcal{Q}_\rho$ is convex, and the optimization over $x\in\mathbb{R}$ can be done in a compact interval since $X\in L^\infty$. As we are considering negative results as losses, we need to correct the sign inside  $\rho(\cdot)$.
\end{proof}

\section{Main properties}\label{sec:prop}
	
 We now explore the main properties of our risk and deviation measures. We begin with those in relation to the position $X$, which is the praxis in risk and deviation measures literature.

	\begin{Prp}\label{prp:sol2}
	Let $R:L^{\infty} \rightarrow\mathbb{R}$ be as in \eqref{averserisk}. Then it has the following properties:
	\begin{enumerate}
		\item Monotonicity: if $X\leq Y$, then $R(X)\geq R(Y), \:\forall \;X,Y \in L^{\infty}$.
		\item Translation Invariance: $R(X+C)=R(X)-C,\:\forall \;X\in L^{\infty},\forall\:C\in\mathbb{R}$.
		\item Positive Homogeneity: $R(\lambda X)=\lambda R(X), \forall \;X\in L^{\infty}, \forall \:\lambda\geq 0$.
		\item Lipschitz continuity. 
		\item Lower semi-continuity regarding $\mathbb{P}-a.s.$ sense if $G,L$ are independent of any $X\in L^\infty$ for any $\mathbb{Q}\in\mathcal{Q}^\prime$. If in addition $\mathcal{Q}^\prime$ is a singleton, then it is continuous in $\mathbb{P}-a.s.$ sense.
		\item Acceptance set: $R(X)=\inf\left\lbrace m\in\mathbb{R}\colon X+m\in \mathcal{A}_{R}\right\rbrace ,\:\forall\:X\in L^\infty$, where  \begin{align*}
\mathcal{A}_{R}&=\{X\in L^{\infty} : R(X)\leq 0\}\\&= \left\lbrace X \in L^\infty:  X\in C_\mathbb{Q}(x)\: \forall \; x<0\:\forall\:\mathbb{Q}\in\mathcal{Q}^\prime \right\rbrace=\bigcap_{\mathbb{Q}\in\mathcal{Q}^\prime}\bigcap_{x<0}C_\mathbb{Q}(x),
		\end{align*}
with $C_\mathbb{Q}(x)=\left( \left\lbrace X\in L^\infty\colon E_\mathbb{Q}[(G+L)1_{\{X\leq x\}}]\geq E_\mathbb{Q}[G]\geq E_\mathbb{Q}[(G+L)1_{\{X< x\}}]\right\rbrace\right) ^c$ 
			\end{enumerate}

\end{Prp}

\begin{proof}
	\begin{enumerate}
		\item Fix $\mathbb{Q}\in\mathcal{Q}^\prime$ and let $g,h\colon L^\infty\times\mathbb{R}\rightarrow\mathbb{R}$ be as $g(X,x)=E_\mathbb{Q}[G1_{\{X\geq x\}}-L1_{\{X< x\}}]$ and $h(X,x)=E_\mathbb{Q}[G1_{\{X>x\}}-L1_{\{X\leq x\}}]$. We have, from the definition of indicator functions, that $h$ and $g$ are non-decreasing in the first argument and non-increasing in the second. Now, let $X,Y\in L^\infty$ with $X\leq Y$. Then $h(X,x)\leq h(Y,x)$ for any $x\in\mathbb{R}$. Furthermore, notice that if $g(X,y)<0$ and $g(X,z)\geq0$, then $z\leq y$. Thus  \[\inf\{x\in\mathbb{R}\colon h(X,x)\leq0\}=\inf\{x\in\mathbb{R}\colon g(X,x)\geq 0,h(X,x)\leq0\}.\] Then, we get from the first order condition of Proposition \ref{prp:prob} that
		\begin{align*}
		x^{*}&=\min\left\lbrace\arg\min\limits_{x\in\mathbb{R}} E_{\mathbb{Q}}[(X-x)^{+}G+(X-x)^{-}L]\right\rbrace=\inf\{x\in\mathbb{R}\colon h(X,x)\leq0\},\\
		y^{*}&=\min\left\lbrace\arg\min\limits_{x\in\mathbb{R}} E_{\mathbb{Q}}[(Y-x)^{+}G+(Y-x)^{-}L]\right\rbrace=\inf\{x\in\mathbb{R}\colon h(Y,x)\leq 0\}.
		\end{align*} Note that such expressions are well defined because, from Proposition \ref{prp:prob}, the argmin set is a closed interval. We then must have $x^{*}\leq y^{*}$ since $\{x\in\mathbb{R}\colon h(Y,x)\leq 0\}\subseteq \{x\in\mathbb{R}\colon h(X,x)\leq 0\}$. By multiplying both sides by $-1$ and taking the supremum over $\mathcal{Q}^\prime$ we get the claim.
		\item Directly obtained from the definition in \eqref{averserisk} by a change of variables.
		\item Analogous to (ii). Note that $R(0)=0$.
		\item  Monotonicity plus Translation Invariance implies Lipschitz continuity regarding $L^\infty$ norm, see Corollary 2 of \cite{delbaen2012monetary}, for instance. In particular, 1 is a valid Lipschitz constant. Thus, we then have for any $X\in L^\infty$ that $|R(X)|=|R(X)-R(0)|\leq\operatorname{ess}\sup |X-0|=\operatorname{ess}\sup |X|<\infty.$ Hence, $R$ is finite and well defined. 
		\item Let $\alpha=\frac{E_\mathbb{Q}[G]}{E_\mathbb{Q}[G+L]}$. Proposition \ref{prp:rob} implies that 	$R(X)=-\inf_{\mathbb{Q}\in\mathcal{Q}^\prime} F_{X,\mathbb{Q}}^{-1} \left( \alpha\right)$ in this case. For $\lim\limits_{n\rightarrow\infty}X_n=X,\:\{X_n\}\subset L^\infty$ we have $F^{-1}_{X,\mathbb{Q}}\left(\alpha\right) = \lim\limits_{n\rightarrow\infty}F^{-1}_{X_n,\mathbb{Q}}\left(\alpha \right)$. Then we get \begin{align*}
		R(X)&=-\inf_{\mathbb{Q}\in\mathcal{Q}^\prime}\lim\limits_{n\rightarrow\infty}F^{-1}_{X_n,\mathbb{Q}}\left(\alpha \right)\\
		&\leq-\lim\limits_{n\rightarrow\infty}\inf\limits_{\mathbb{Q}\in\mathcal{Q}^\prime}F^{-1}_{X_n,\mathbb{Q}}\left(\alpha \right) \leq\liminf\limits_{n\rightarrow\infty}R(X_n).
		\end{align*} When $\mathcal{Q}^\prime=\{\mathbb{Q}\}$ we have \[R(X)=- \lim\limits_{n\rightarrow\infty}F_{X_n,\mathbb{Q}}^{-1} \left( \alpha\right)=\lim\limits_{n\rightarrow\infty}R(X_n).\]
		\item From the properties of $R$, we have that it can be recovered by $\mathcal{A}_{R}=\{X\in L^{\infty} : R(X)\leq 0\}$, see Proposition 4.6 of \cite{Follmer2016}, for instance. We then have that
	\begin{align*}
	\mathcal{A}_{R}	&= \left\lbrace X \in L^\infty:\min\left\lbrace\arg\min\limits_{x\in\mathbb{R}} E_{\mathbb{Q}}[(X-x)^{+}G+(X-x)^{-}L]\right\rbrace \geq 0,  \forall \;  \mathbb{Q}\in\mathcal{Q}^\prime    \right\rbrace 
	\\ 
	&= \bigcap_{\mathbb{Q}\in\mathcal{Q}^\prime}\left\lbrace X \in L^\infty: E_\mathbb{Q}[(G+L)1_{\{X\leq x\}}]\geq E_\mathbb{Q}[G]\geq E_\mathbb{Q}[(G+L)1_{\{X< x\}}] \Rightarrow x \geq 0\right\rbrace \\ 
	&=\left\lbrace X \in L^\infty:  X\in C_\mathbb{Q}(x)\: \forall \; x<0\:\forall\:\mathbb{Q}\in\mathcal{Q}^\prime \right\rbrace=\bigcap_{\mathbb{Q}\in\mathcal{Q}^\prime}\bigcap_{x<0}C_\mathbb{Q}(x)
	\end{align*}
	\end{enumerate}
\end{proof}

Monotonicity requires that if one position generates worse results than another, then its risk will be greater. Translation Invariance ensures that if one adds a certain gain to a position, then its risk will decrease by the same amount. Positive Homogeneity indicates that risk proportionally increases with position size. In this case, $\mathcal{A}_{R}$ is a closed  cone that contains $L_+^{\infty}$ and has no intersection with  $\{X\in L^{\infty} :X < 0\}$. In the context of Proposition \ref{prp:rob}, we have that
\[
\mathcal{A}_{R}=\left\lbrace X\in L^{\infty} : \mathbb{Q}(X<0)\leq\frac{E_\mathbb{Q}[G]}{E_\mathbb{Q}[G+L]},\:\forall\:\mathbb{Q}\in\mathcal{Q}^\prime\right\rbrace.\]
This set is surplus invariant in the sense of \cite{Koch-Medina2017}, i.e., $X\in\mathcal{A}_{R}$ and $X^-\geq Y^-$ implies in $Y\in\mathcal{A}_{R}$. To see this, note that \[\mathbb{Q}(Y\leq0)=\mathbb{Q}(Y^-\geq0)\leq\mathbb{Q}(X^-\geq0)=\mathbb{Q}(X\leq0)\leq\frac{E_\mathbb{Q}[G]}{E_\mathbb{Q}[G+L]},\:\forall\:\mathbb{Q}\in\mathcal{Q}^\prime.\] Thus $Y\in\mathcal{A}_{R}$. 

\begin{Rmk}
	Regarding $\mathbb{P}-a.s.$ lower semi-continuity, it could be obtained if the argmin sets were singleton for any $X\in L^\infty$. For this convergence of argmin sets see Theorem 7.33 in \cite{Rockafellar2009}, for instance. We also have that the maximum of argmin sets is monotone in $X$. The reasoning is similar to that in (i) but considering $g\colon L^\infty\times\mathbb{R}\rightarrow\mathbb{R}$ as $g(X,x)=E_\mathbb{Q}[G1_{\{X\geq x\}}-L1_{\{X< x\}}]$ and $\sup\{x\in\mathbb{R}\colon g(X,x)\geq0\}$.
\end{Rmk}

	\begin{Prp}\label{prp:sol}
		Let $RD:L^{\infty} \rightarrow\mathbb{R}_+$ be as in \eqref{aversedev}. Then it has the following properties:
\begin{enumerate}
		\item Translation Insensitivity: $RD(X+C)=RD(X),\forall \; X\in L^{\infty}, \forall \; C \in \mathbb{R}$.
	\item Positive Homogeneity: $RD(\lambda X)=\lambda RD(X),\:\forall \; X\in L^{\infty},\forall \; \lambda \geq 0$.
	\item Non-negativity: For all $X\in L^{\infty}$, $RD(X)=0$ for constant $X$, and $RD(X)>0$ for non-constant $X$.
	\item Convexity: $RD(\lambda X+(1-\lambda)Y)\leq \lambda RD(X)+(1-\lambda)RD(Y),\:\forall \; X,Y\in L^{\infty},\:\forall\:\lambda\in[0,1]$.

	\item Lipschitz continuity.
	\item Lower semi-continuity regarding $\mathbb{P}-a.s.$ sense for bounded sequences. If in addition $\mathcal{Q}^\prime$ is a singleton, then it is continuous in $\mathbb{P}-a.s.$ sense for bounded sequences.
		\item Range Dominance: \begin{itemize}
		\item if $G\leq 1$, then $RD(X)\leq\sup\limits_{\mathbb{Q}\in\mathcal{Q}^\prime}E_\mathbb{Q}[X]-\operatorname{ess}\inf X,\forall \; X\in L^{\infty}$. 
		\item if $L\leq 1$, then $RD(X)\leq\operatorname{ess}\sup X-\sup\limits_{\mathbb{Q}\in\mathcal{Q}^\prime}E_\mathbb{Q}[X],\forall \; X\in L^{\infty}$. 
		\item  if $G,L\leq 1$, then $RD(X)\leq\dfrac{\operatorname{ess}\sup X-\operatorname{ess}\inf X}{2},\forall \; X\in L^{\infty}$.
	\end{itemize}
	\item Dual representation: $RD(X)=\sup\limits_{Z\in\mathcal {Z}_{\mathcal{Q}^\prime}}E[XZ]$, 
	where $\mathcal {Z}_{\mathcal{Q}^\prime}$ is the closed convex hull of \[\cup_{\mathbb{Q}\in\mathcal{Q}^\prime}\left\lbrace Z\in L^1:E[Z]=0,\:E[XZ]\leq  E_{\mathbb{Q}}[(X-x)^{+}G+(X-x)^{-}L]\:\forall\:x\in\mathbb{R},\:\forall\:X\in L^\infty\right\rbrace.\] 
	\end{enumerate}
\end{Prp}
	
\begin{proof}
\begin{enumerate}
	
	\item Directly obtained by a change of the minimization variable.
	
	\item Analogous to (i).

	\item Since $R(C)=-C,\:\forall\:C\in\mathbb{R}$, we have that $RD(C)=0$. If $X$ is not a constant, with abuse of notation, we have that $\exists\:\omega\in\Omega$ with $X(\omega)\neq -R(X)$. We then get that at least one between $(X(\omega)+R(X))^{+}$ or $(X(\omega)+R(X))^{-}$ is strictly positive. Hence, $RD(X)>0$.

		\item Consider any pair $X,Y\in L^\infty$ and any $x\in\mathbb{R}$. We then obtain for each $\mathbb{Q}\in\mathcal{Q}^\prime$ that
		\begin{align*}
		&\min\limits_{x_1+x_2\in\mathbb{R}}E_{\mathbb{Q}}\left[\left(X+Y-(x_1+x_2)\right)^{+}G+\left(X+Y-(x_1+x_2)\right)^{-}L\right]\\
		\leq&\min\limits_{x_1,x_2\in\mathbb{R}}\left\lbrace E_{\mathbb{Q}}\left[\left(X-x_1\right)^{+}G+\left(X-x_1\right)^{-}L+\left(Y-x_2\right)^{+}G+\left(Y-x_2\right)^{-}L\right] \right\rbrace\\
		=&\min\limits_{x_1\in\mathbb{R}} E_{\mathbb{Q}}\left[\left(X-x_1\right)^{+}G+\left(X-x_1\right)^{-}L\right] +		\min\limits_{x_2\in\mathbb{R}} E_{\mathbb{Q}}\left[\left(Y-x_2\right)^{+}G+\left(X-x_2\right)^{-}L\right].
		\end{align*}
	By taking the supremum over $\mathbb{Q}\in\mathcal{Q}^\prime$ on both sides of the resulting inequality and from linearity of the real line one gets $RD(X+Y)\leq RD(X)+RD(Y),\forall \; X,Y\in L^{\infty}$. Together with Positive Homogeneity (item (ii)) it implies that $RD$ is a convex functional. 
		
	\item  For Lipschitz continuity, note that $RD$ is bounded since for any $X\in L^\infty$ we have \[\left|RD(X)\right|\leq \sup_{\mathbb{Q}\in\mathcal{Q}^\prime}E_\mathbb{Q}\left[X^+G+X^-L\right]\leq \operatorname{ess}\sup_\mathbb{P}(\max(G,L))\operatorname{ess}\sup_\mathbb{P}|X|<\infty.\]
Furthermore, from both Convexity and Positive Homogeneity, we have that $RD$ is sub-linear, i.e.,  $RD(X+Y)\leq RD(X)+RD(Y)),\:\forall\:X,Y\in L^\infty $. Thus, \[RD(X)-RD(Y)\leq RD(X-Y)\leq\operatorname{ess}\sup_\mathbb{P}(\max(G,L))\operatorname{ess}\sup_\mathbb{P}|X-Y|.\] By reversing the roles of $X$ and $Y$ we get  \[|RD(X)-RD(Y)|\leq\operatorname{ess}\sup_\mathbb{P}(\max(G,L))\operatorname{ess}\sup_\mathbb{P}|X-Y|.\] Since $G,L\in L^\infty_{++}$ are fixed, we get the claim.

\item  For any $x\in\mathbb{R}$ and any $\mathbb{Q}\in\mathcal{Q}^\prime$ by Dominated convergence we have that $\lim\limits_{n\rightarrow\infty}X_n=X,\:\{X_n\}\subset L^\infty$ bounded implies, for any bounded sequence $\{x_n\}\subset\mathbb{R}$ such that $\lim\limits_{n\rightarrow\infty}x_n=x$, in \[E_{\mathbb{Q}}[(X-x)^{+}G+(X-x)^{-}L]=\lim\limits_{n\rightarrow\infty}E_{\mathbb{Q}}[(X_n-x_n)^{+}G+(X_n-x_n)^{-}L].\] Let $\{x^{*}_{\mathbb{Q},n}\}$ be a sequence where each member is from the argmin set for $X_n$ under $\mathbb{Q}$. Since $\{X_n\}$ is bounded, also is bounded $\{x^{*}_{\mathbb{Q},n}\}$ for any $\mathbb{Q}\in\mathcal{Q}^\prime$. By the Bolzano-Weierstrass Theorem, we have, by taking a subsequence if needed, that $x^{*}_\mathbb{Q}=\lim\limits_{n\rightarrow\infty}x^{*}_{\mathbb{Q},n}$ is well defined and finite. We then get that \begin{align*}
RD(X)&\leq \sup\limits_{\mathbb{Q}\in\mathcal{Q}^\prime}E_{\mathbb{Q}}[(X-x^{*}_\mathbb{Q})^{+}G+(X-x^{*}_\mathbb{Q})^{-}L]\\
&=\sup\limits_{\mathbb{Q}\in\mathcal{Q}^\prime}\lim\limits_{n\rightarrow\infty}E_{\mathbb{Q}}[(X_n-x^{*}_{\mathbb{Q},n})^{+}G+(X_n-x^{*}_{\mathbb{Q},n})^{-}L]\\
&\leq\liminf\limits_{n\rightarrow\infty}\sup\limits_{\mathbb{Q}\in\mathcal{Q}^\prime}E_{\mathbb{Q}}[(X_n-x^{*}_{\mathbb{Q},n})^{+}G+(X_n-x^{*}_{\mathbb{Q},n})^{-}L]\\
&=\liminf\limits_{n\rightarrow\infty}RD(X_n).
\end{align*}
When $\mathcal{Q}^\prime=\{\mathbb{Q}\}$ we have in addition that
\begin{align*}
RD(X)&=\min\limits_{x\in\mathbb{R}}\lim\limits_{n\rightarrow\infty}E_{\mathbb{Q}}[(X_n-x)^{+}G+(X_n-x)^{-}L]\\
&\geq\limsup\limits_{n\rightarrow\infty}\min\limits_{x\in\mathbb{R}}E_{\mathbb{Q}}[(X_n-x)^{+}G+(X_n-x)^{-}L]\\
&=\limsup\limits_{n\rightarrow\infty}RD(X_n).
\end{align*}
Hence $RD(X)=\lim\limits_{n\rightarrow\infty}RD(X_n)$.

\item For $G\leq1$ we have for any $X\in L^\infty$ that \[RD(X)\leq\sup\limits_{\mathbb{Q}\in\mathcal{Q}^\prime}E_\mathbb{Q}\left[\left(X-\operatorname{ess}\inf_\mathbb{P}X\right)G\right]\leq\sup\limits_{\mathbb{Q}\in\mathcal{Q}^\prime}E_\mathbb{Q}[X]-\operatorname{ess}\inf_\mathbb{P} X.\] For $L\leq 1$ the reasoning is analogous with $(\operatorname{ess}\sup_\mathbb{Q}X-X)L$ rather than $(X-\operatorname{ess}\inf_\mathbb{Q}X)G$. When both $G,L\leq 1$ the upper bound is a direct consequence.

\item Let $RD_\mathbb{Q}(X)=\min\limits_{x\in\mathbb{R}} E_{\mathbb{Q}}[(X-x)^{+}G+(X-x)^{-}L],\:\forall\:\mathbb{Q}\in\mathcal{Q}^\prime$. From the properties of $RD$, Theorem 1 of  \cite{rockafellar2006generalized} and The Main Theorem of \cite{pflug2006subdifferential}, for instance, we have the following dual representation:
\begin{equation*}
RD(X)=\sup\limits_{Z\in\mathcal {Z}_{\mathcal{Q}^\prime}}E[XZ],
\end{equation*}
where \begin{align*}
\mathcal{Z}_{\mathcal{Q}^\prime}&=\left\lbrace Z\in L^1:E[Z]=0,E[XZ]\leq RD(X),\:\forall\:X\in L^\infty\right\rbrace\\
&=\left\lbrace Z\in L^1:E[Z]=0,\exists\:\:\mathbb{Q}\in\mathcal{Q}^\prime\:\text{s.t.}\:E[XZ]\leq RD_{\mathbb{Q}}(X),\:\forall\:X\in L^\infty\right\rbrace\\
&=\cup_{\mathbb{Q}\in\mathcal{Q}^\prime}\left\lbrace Z\in L^1:E[Z]=0,\:E[XZ]\leq RD_{\mathbb{Q}}(X),\:\forall\:X\in L^\infty\right\rbrace. 
\end{align*} The fact that supremum is not altered by considering the closed convex hull is due to convexity and lower semi-continuity of $RD$ and can be verified in  Proposition 4.14 of \cite{Righi2018b}. 
\end{enumerate}
 
\end{proof}

Translation Insensitivity indicates the deviation in relation to the expected value does not change if we add a constant value. Convexity, which is related to the principle of diversification, implies that the risk of a combined position is less than the combination of individual risks. Non-negativity assures that there is dispersion only for non-constant positions. Range Dominance implies sharp bounds for the functional. Such properties make $RD$ a generalized deviation measure in the sense of  \cite{rockafellar2006generalized} and \cite{Pflug2007}. Under the special conditions of Proposition \ref{prp:rob}, we have that $\mathcal{Z}_{\mathcal{Q}^\prime}$ is the closed convex hull of $\cup_{\mathbb{Q}\in\mathcal{Q}^\prime}\mathcal{Z}_{\mathbb{Q}}$, where \[\mathcal {Z}_\mathbb{Q}=\left\lbrace Z\in L^1:Z=\left(1-\frac{d\mathbb{Q}^{'}}{d\mathbb{Q}}\right)E[G],\mathbb{Q}^{'}\ll\mathbb{Q},\frac{d\mathbb{Q}^{'}}{d\mathbb{Q}}\leq\left( \frac{E_\mathbb{Q}[G]}{E_\mathbb{Q}[G+L]}\right) ^{-1}\right\rbrace.\]

Other very relevant properties in risk/deviation measures literature are, for $f\colon L^\infty\rightarrow\mathbb{R}$, Law Invariance:  if $F_X=F_Y$, then $f(X)=f(Y)$; and Co-monotonic Additivity: $f(X+Y)=f(X)+f(Y)$ for $X,Y$ co-monotone, i.e., $(X(\omega)-X(\omega^\prime))(Y(\omega)-Y(\omega^\prime))\geq0,\:\forall\:\omega,\omega^\prime\in\Omega$. Such properties are directly related to the notion of Spectral and Distortion risk/deviation measures, see \cite{Acerbi2002c} and \cite{Grechuk2009}, for instance. The intuition is that functionals depend solely on the distribution of random variables and there is no diversification benefit in the case of extreme positive association. Concerning to Co-monotonic Additivity, when $\mathcal{Q}^\prime=\{\mathbb{Q}\}$ is a singleton and $X$ is $\mathbb{Q}$-independent of $G$ and $L$, it is satisfied by $R$ and $RD$ because both VaR and ES possess this property.

Regarding Law Invariance, it is not satisfied  in general for our approach since both $R$ and $RD$ depend on the joint distributions of pairs $X,G$ and $X,L$. Matter of fact, we would have to adapt the property for Joint Law Invariance as $F_{X,G}=F_{Y,G}$ and $F_{X,L}=F_{Y,L}$ implying into $f(X)=f(Y)$, where $F_{X,Y}$ is the joint distribution of $X$ and $Y$. An exception is the situation of independence in Proposition \ref{prp:rob}.  Moreover, our approach under a set of probability measures is linked to the $\mathcal{Q}^\prime$-Law Invariance ($\mathcal{Q}^\prime$-based) of \cite{Wang2018} and \cite{Righi2018b} where if $F_{X,\mathbb{Q}}=F_{Y,\mathbb{Q}},\:\forall\:\mathbb{Q}\in\mathcal{Q}^\prime$, then $f(X)=f(Y)$. Nonetheless, even in such case, we would have to adapt for the joint distributions issue. We let this pursue for future research since it is outside the scope of this paper.

We now focus on the effect of costs $G$ and $L$ in our risk and deviation measures. This is important since the minimization procedure is directly dependent on such costs, beyond the financial position $X$. Note that we are in fact considering the maps defined as $G\rightarrow R_{G,L,\mathcal{Q}^\prime}(X)$, $L\rightarrow R_{G,L,\mathcal{Q}^\prime}(X)$, $G\rightarrow RD_{G,L,\mathcal{Q}^\prime}(X)$ and $L\rightarrow RD_{G,L,\mathcal{Q}^\prime}(X)$, for some fixed $X\in L^\infty$. 

\begin{Prp}\label{prp:GL}
Let $R:=R_{G,L,\mathcal{Q}^\prime}:L^{\infty} \rightarrow\mathbb{R}$ be as in \eqref{averserisk}. Then it has the following properties for any $X\in L^\infty$:
\begin{enumerate}
	\item Non-increasing in $G$ and non-decreasing in $L$.
	\item Boundedness in both $G$ and $L$.
	\item Lower semi-continuity in both norm and $\mathbb{P}-a.s.$ sense for both $G$ and $L$  if $F_X$ is continuous and $F_{X,\mathbb{Q}}$ strictly increasing for any $\mathbb{Q}\in\mathcal{Q}^\prime$. If $\mathcal{Q}^\prime$ is a singleton, then it is continuous in both norm and $\mathbb{P}-a.s.$ sense. 
\end{enumerate}
\end{Prp}

\begin{proof}
 Consider the functions $f_\mathbb{Q},g_\mathbb{Q},h_\mathbb{Q}\colon L^\infty\times L^\infty\times\mathbb{R}\rightarrow\mathbb{R}$, respectively defined as $f_\mathbb{Q}(G,L,x)=E_\mathbb{Q}[G1_{\{X\geq x\}}-L1_{\{X\leq x\}}]$, $g_\mathbb{Q}(G,L,x)=E_\mathbb{Q}[G1_{\{X\geq x\}}-L1_{\{X< x\}}]$ and $h_\mathbb{Q}(G,L,x)=E_\mathbb{Q}[G1_{\{X>x\}}-L1_{\{X\leq x\}}]$. We then have:
\begin{enumerate}
	\item Similar to Monotonicity property for $X$ (item (i)) in Proposition \ref{prp:sol2} by noticing that, for any $\mathbb{Q}\in\mathcal{Q}^\prime$, both functions $g_\mathbb{Q}(G,L,x)$ and $h_\mathbb{Q}(G,L,x)$ are non-decreasing in the first argument and non-increasing in the second and third arguments.
 \item For fixed $X\in L^\infty$ we have $R_{G,L,\mathcal{Q}^\prime}(X)\in[\operatorname{ess}\inf X,\operatorname{ess}\sup X]$ for any $G,L\in L^\infty_{++}$. Thus, $|R_{G,L,\mathcal{Q}^\prime}(X)|\leq\operatorname{ess}\sup |X|<\infty$.
\item  We focus on the proof for $G$. Concerning to $L$, the reasoning is analogous. In this case, we have from Proposition \ref{prp:prob} that argmin sets are singleton. We begin with lower semi-continuity in $\mathbb{P}-a.s.$ sense. Let $A_{G,\mathbb{Q}}=\{x\in\mathbb{R}\colon f_\mathbb{Q}(G,L,x)=0\}$, which, from Proposition \ref{prp:prob}, is a singleton since $F_X$ is continuous and $F_{X,\mathbb{Q}}$ is strictly increasing by hypothesis. Moreover, we have that $\{G_n\}\subset L^\infty_{++}$ and $\lim\limits_{n\rightarrow\infty}G_n=G$, which implies in $\lim\limits_{n\rightarrow\infty}A_{G_n,\mathbb{Q}}=A_{G,\mathbb{Q}}$. Again, such convergence of $\arg\min$ sets is guaranteed by Theorem 7.33 in \cite{Rockafellar2009}, for instance. We then obtain that
\begin{align*}
R_{G,L,\mathcal{Q}^\prime}(X)&=\sup_{\mathbb{Q}\in\mathcal{Q}^\prime}\left\lbrace -\min \left\lbrace \lim\limits_{n\rightarrow\infty}A_{G_n,\mathbb{Q}}\right\rbrace\right\rbrace\\
&\leq\liminf\limits_{n\rightarrow\infty}\sup_{\mathbb{Q}\in\mathcal{Q}^\prime}\left\lbrace -\min A_{G_n,\mathbb{Q}}\right\rbrace\\
&=\liminf\limits_{n\rightarrow\infty}R_{G_n,L,\mathcal{Q}^\prime}(X).
\end{align*} 
When $\mathcal{Q}^\prime=\{\mathbb{Q}\}$ we have that \[R_{G,L,\mathcal{Q}^\prime}(X)=-\min \left\lbrace \lim\limits_{n\rightarrow\infty}A_{G_n,\mathbb{Q}}\right\rbrace=- \lim\limits_{n\rightarrow\infty}A_{G_n,\mathbb{Q}}=\lim\limits_{n\rightarrow\infty}R_{G_n,L,\mathcal{Q}^\prime}(X).\]
Since convergence in $L^\infty$ norm implies convergence in $\mathbb{P}-a.s.$ sense, we have the claim.
\end{enumerate}
\end{proof}

The monotone behavior for $G$ and $L$ is a consequence of the fact that more or less weight is put to overestimation or underestimation, respectively. When argmin sets are singleton, continuity properties imply that $R(X)\uparrow -\operatorname{ess}\inf X$ when $G\downarrow0$  representing the worst possible loss for $X$. On the other side, we have that it assumes the value $R(X)\downarrow- \operatorname{ess}\sup X$  when $L\downarrow0$ representing the best possible loss for $X$, corroborating to the identified monotone behavior. 

\begin{Prp}\label{prp:GL2}
	Let $RD:=RD_{G,L,\mathcal{Q}^\prime}:L^{\infty} \rightarrow\mathbb{R}$ be as in \eqref{aversedev}. Then it has the following properties for any $X\in L^\infty$:
	\begin{enumerate}
		\item Non-decreasing in both $G$ and $L$.
		\item Concavity for both $G$ and $L$ if $\mathcal{Q}^\prime$ is a singleton. Furthermore, we have  \begin{align*}
		RD_{(\lambda G_1+(1-\lambda )G_2),L,\mathcal{Q}^\prime}(X)\geq\max\limits_{\mu\in[0,1]}\left\lbrace \mu\lambda RD_{G_1,L,\mathcal{Q}^\prime}(X)+(1-\mu)(1-\lambda)RD_{G_2,L,\mathcal{Q}^\prime}(X)\right\rbrace,\\
		RD_{G,(\lambda L_1+(1-\lambda )L_2),\mathcal{Q}^\prime}(X)\geq\max\limits_{\mu\in[0,1]}\left\lbrace \mu\lambda RD_{G,L_1,\mathcal{Q}^\prime}(X)+(1-\mu)(1-\lambda)RD_{G,L_2,\mathcal{Q}^\prime}(X)\right\rbrace,
		\end{align*}
		$\forall\:\lambda\in[0,1]$, if $\mathcal{Q}^\prime$ is a convex set. 
		\item Norm and $\mathbb{P}-a.s.$ lower semi-continuity for bounded sequences in both $G$ and $L$. If in addition $\mathcal{Q}^\prime$ is a singleton, then it is continuous in norm and $\mathbb{P}-a.s.$ sense for bounded sequences.
	\end{enumerate}
\end{Prp}

\begin{proof}
 We always focus on the case $G$. The reasoning for $L$ is quite similar. We then have:
	\begin{enumerate}
		\item If $G_1\geq G_2$, $G_1,G_2\in L^\infty_{++}$, 
		then \[E_\mathbb{Q}[(X-x)^+G_1+(X-x)^-L]\geq E_\mathbb{Q}[(X-x)^+G_2+(X-x)^-L],\:\forall\:x\in\mathbb{R},\:\forall\:\mathbb{Q}\in\mathcal{Q}^\prime.\] Since both minimum and supremum are monotone functions, we have non-decreasing behavior. 
		\item Let $G=\lambda G_1 +(1-\lambda) G_2$ for $G_1,G_2\in L^\infty_{++}$, $\lambda\in[0,1]$ and $\mathcal{Q}^\prime=\{\mathbb{Q}\}$. We then have that
	\begin{align*}
&RD_{G,L,\mathcal{Q}^\prime}(X)\\
\geq&\lambda\min\limits_{x\in\mathbb{R}}E_\mathbb{Q}\left[(X-x)^+G_1+(X-x)^-L\right]+(1-\lambda)\min\limits_{x\in\mathbb{R}}E_\mathbb{Q}\left[(X-x)^+G_2+(X-x)^-L\right]\\
=&\lambda RD_{G_1,L,\mathcal{Q}^\prime}(X)+(1-\lambda)RD_{G_2,L,\mathcal{Q}^\prime}(X). 
	\end{align*} 
Let $\mathcal{Q}^\prime$ be convex. We have that $f\colon L^\infty\times\mathcal{Q}^\prime\rightarrow\mathbb{R}$ defined as \[f(G,\mathbb{Q})=\min\limits_{x\in\mathbb{R}}E_\mathbb{Q}\left[(X-x)^+G+(X-x)^-L\right]\] is a bi-concave (concave in both arguments) function. We thus get for $G=\lambda G_1+(1-\lambda)G_2$ with $G_1,G_2\in L^\infty_{++}$ and $\lambda\in[0,1]$ that
\begin{align*}
RD_{G,L,\mathcal{Q}^\prime}(X)&\geq \lambda\mu f(G_1,\mathbb{Q}_1)+\lambda(1-\mu)f(G_1,\mathbb{Q}_2)+(1-\lambda)\mu f(G_2,\mathbb{Q}_1)\\
&+(1-\lambda)(1-\mu)f(G_2,\mathbb{Q}_2),\:\forall\:\mathbb{Q}_1,\mathbb{Q}_2\in\mathcal{Q}^\prime,\:\forall\:\mu\in[0,1]\\
&\geq\lambda\mu f(G_1,\mathbb{Q}_1)+(1-\lambda)(1-\mu)f(G_2,\mathbb{Q}_2),\:\forall\:\mathbb{Q}_1,\mathbb{Q}_2\in\mathcal{Q}^\prime,\:\forall\:\mu\in[0,1].
\end{align*}
Thus by taking the supremum over pairs $\mathbb{Q}_1,\mathbb{Q}_2\in\mathcal{Q}^\prime$ we get \[RD_{G,L,\mathcal{Q}^\prime}(X)\geq\lambda\mu[RD_{G_1,L,\mathcal{Q}^\prime}(X)+(1-\lambda)(1-\mu)RD_{G_2,L,\mathcal{Q}^\prime}(X),\:\forall\:\mu\in[0,1].\] By taking the supremum over $\mu\in[0,1]$, which is in fact attained, we get the claim.
\item  
Lower semi-continuity regarding $\mathbb{P}-a.s.$ sense follows similar steps to item (vi) in Proposition \ref{prp:sol} by considering $\{G_n\}$ instead of $\{X_n\}$. Since convergence in $L^\infty$ norm implies convergence in $\mathbb{P}-a.s.$ sense, we have the claim.
	\end{enumerate}
\end{proof}

For $RD$ we have that larger values are observed when both costs rise since both $G$ and $L$ are punitive. The concave behavior intuitively means that marginal costs are non-increasing, which is well sounded in practical matters. Regarding continuity, we have when both extremes cases $G\downarrow0$ and $L\downarrow0$  occurs that $RD(X)\downarrow0$, in consonance with the pattern for $R$.

\section{Empirical illustration}\label{sec:illust}

In this section, we expose an empirical illustration of our proposed approach for capital determination problem considering real financial data.
We consider $\mathcal{Q}^\prime = \mathcal{Q}_{\rho}$, where $\rho$ refers to EL, MSD, ES, EVaR and ML.
These measures are described in Example \ref{ex:risk}.
For MSD, we chose $\beta =1$ to incorporate all the deviation term\footnote{Similar value of $\beta$ is used in \cite{righi2017simulation} to estimate loss-deviation risk measures. Their results showed a better performance of risk measures penalized by deviation in comparison with their counterpart without deviation.}. 
The values of $\alpha$ are  $0.025$ in the case of ES and $0.00145$ for EVaR.  The latest revisions of the Committee on Banking Supervision recommend this  $\alpha$  for ES (see \cite{basel2013fundamental}) and 
for EVaR this value  is closely comparable to $ES^{0.025}$ for a normally distributed $X$ (see \cite{Bellini2017}).

For the financial position $X$, we consider log-returns of  S\&P 500 U.S. market index  multiplied by 100, which is a typical example of  a financial asset used in academic research. 
Regarding the costs of risk overestimation ($G$) and underestimation ($L$), we consider yield rates of the U.S. Treasury Bill with maturity of three months and the U.S. Dollar based Overnight London Interbank Offered Rate (LIBOR)\footnote{We emphasize that here our intention is only to illustrate the applicability of our approach. In this way, different costs from risk overestimation ($G$) and underestimation ($L$) can be considered. In 
periods of greater instability, more aggressive rates can be used as a cost from underestimation, for example.}. These rates represent, respectively, a risk-free investment with high liquidity where the surplus over capital requirement could be safely applied, and a rate for providential loans when capital requirement is not enough. We consider daily data from January 2001 to May 2018, being $N = 4376$ observations. We convert both yield rates to daily frequency.

In this estimation process of risk measures, we have log-returns defined in some discrete probability space $\Omega=\left(\omega_1,\cdots,\omega_n \right)$, as $X(\omega_i)=X_i,\:i=1,\cdots, n$, where $n$ is the number of observations. For this approach, we consider  $\mathbb{P}(X=X_i)=\mathbb{P}(\omega_i)=\frac{1}{n}$. This leads to the empirical distribution and expectation, respectively, defined as:
\begin{equation*}
F_X(x)=\frac{1}{n}\sum_{i=1}^{n}1_{\{X_i\leq x\}}, \:E[X]=\frac{1}{n}\sum_{i=1}^{n}X_i.
\end{equation*} This empirical method of estimation, known as historical simulation (HS), is a non-parametric approach that makes virtually no  assumptions about the data. Note that $\mathcal{Q}_\rho$ is considered for each choice of $\rho$ as exposed in Example 1. In this case $d\mathbb{Q}/d\mathbb{P}\in\mathbb{R}^n$ with $\frac{d\mathbb{Q}}{d\mathbb{P}}(\omega_i)=n\mathbb{Q}(\omega_i),\:i=1,\cdots,n$. Note that in this case any probability $\mathbb{Q}$ is absolutely continuous in relation to $\mathbb{P}$. For instance, we have $\mathcal{Q}_{ES^\alpha}=\left\lbrace x=(x_1,\cdots,x_n)\in[0,1]^n\colon\sum_{i=1}^nx^i=1,\:x_i\leq\frac{1}{n\alpha}\:\forall\:i\in\{1,\cdots,n\}\right\rbrace $.  Under this specification we adapt \eqref{averserisk} and compute $R_\rho$ as: \begin{equation}
R_{\rho}(X)=\max\limits_{\mathbb{Q}\in\mathcal{Q}_\rho}\left\lbrace -\min\left\lbrace \arg\min\limits_{x\in\mathbb{R}}\sum_{i=1}^{N}\{ [(X_i-x)^{+}G_i+(X_i-x)^{-}L_i]\mathbb{Q}(\omega_i)\}\right\rbrace\right\rbrace \label{eq:discrete1}. 
\end{equation}

In this case, the supremum over $\mathcal{Q}_\rho$ is attained since it is compact because $\Omega$ has finite dimension. We compare the results of our risk measures with those obtained by risk measures using the following loss functions  $-xG+(X-x)^{-}$ and $-xG+(X-x)^{-}L$, which we name $R_{\rho}^{b}$ and $R_{\rho}^{c}$, respectively.  
These measures are in our illustration because they aim to identify the optimal amount of capital minimizing the costs that link to risk underestimation and overestimation. 
Although $R_{\rho}^{b}$ do not explicitly consider $L$, it is a particular case of $-xG+(X-x)^{-}L$  when $L = 1$. 
We compute $R_{\rho}^{b}$ and $R_{\rho}^{c}$ by procedure exposed on  formulation \eqref{eq:discrete1}, changing the loss function. 
We also compare our risk measures with those usually considered for capital requirements, which are  $VaR^{0.01}$ and $ES^{0.025}$, and with ML\footnote{The reader should not confuse our risk measures computed under the dual set of ES and ML, $\mathcal{Q}_{ES^{\alpha}}$ and $\mathcal{Q}_{ML}$, which generate $R_{ES}(X)$ and $R_{ML}(X)$, and the Expected Shortfall and Maximum Loss of $X$, $ES^\alpha(X)$ and $ML(X)$, respectively.}. 
Although VaR is not a coherent risk measure, it is the most common risk measure currently used in the literature and industry. 

In Figure \ref{fig:plot_datarm}, we exhibit the graphic illustration of  series described and the probability for quantile of $X$ represented  by  $\frac{G}{G+L}$  for $R_{\rho}$, $G$  for $R_{\rho}^{b}$, and $\frac{G}{L}$ for $R_{\rho}^{c}$, where $\rho=EL$. Such quantities would be respective solutions in the fashion of Proposition \ref{prp:rob} for such loss function and can be considered for benchmarking. We name them, respectively, as  Quantile$_1$, Quantile$_2$, and  Quantile$_3$. 
One can note this sample contains both turbulent and calm periods, as visualized by the volatility clusters on log-returns, peaks and bottoms on the price series. Regarding the yield rates, there is a huge change in their dynamics at the end of 2008, possibly due to economic changes generated by the sub-prime crisis. To isolate the two distinct patterns identified for the costs, we divided the sample into two periods (2001--2008 and 2009--2018). So, we expose the results for the whole sample and for the two sub-samples. We consider a rolling estimation window of 250 observations (around one year of business days) to compute risk measures, excluding the year of 2001 for this out-sample forecasting exercise. In this sense, for each day in the out-sample period, we use the last 250 observations to compute the empirical distribution and risk measures.

The change in the dynamics of costs is also observed when analyzing the graphical illustration of probability distribution. 
 Regarding Quantile$_1$,  in the first sub-sample, there is a stable evolution around 46\%, close to the middle of the distribution function. In the second sub-sample, there is strong variation with values representing smaller probabilities linked to more extreme losses, which affects the value of our risk measures, but not necessarily VaR and ES. Hence, the sample split we make is an interesting feature of our analysis. 
 For the Quantile$_2$ the evolution coincides with the daily rates of $G$, which has values less than $1\%$, i.e., quantile used in the computation of VaR.
Thus, for these measures, we expect more extreme losses compared to the values computed by VaR and ES. 

In relation to probability distribution of  Quantile$_3$, we note some particularities.
In the first sub-sample the variation of probability it fluctuates around 89$\%$, while in the second sub-sample is around 57$\%$. However, in both sub-samples  there are values greater than 100$\%$, which contradicts the expected values for a probability. 
We identify these values because our structure does not impose  $G \leq L$. For the purpose of insurance  $R_{\rho}^c$, as originally proposed, becomes an interesting alternative since the values of its probability distribution are closer to the upper tail compared to the lower tail. However, our interest is a risk measurement procedure that minimizing the mentioned costs for the determination of capital.

We expose in Figure \ref{fig:plot_res} a plot with the time series of all estimated risk measures in relation to the negative of log-returns. We make this sign conversion because risk measures have their values in terms of losses. Results indicate risk measures estimates follow the evolution of losses on the financial position, as expected. Periods of higher volatility and losses exhibit larger values for risk measures. Complementing, in Table \ref{tab:res}, we expose some descriptive statistics of these series, and the cost criteria (sum of daily costs in the sample) linked to the risk measures, which we obtain in the following manner:
\begin{align}
&\text{Cost} = \sum_{t=1}^{T}[(X_t - x_t)^{+}G_t+(X_t-x_t)^{-}L_t], \nonumber\\
&\text{Cost}^b= \sum_{t=1}^{T} [-x_tG_t+(X_t-x_t)^{-}], \nonumber\\
&\text{Cost}^c = \sum_{t=1}^{T}[-x_tG_t+(X_t-x_t)^{-}L_t],\label{fmeas}
\end{align}
 where $x_t$ is the predicted value in period $t$ with corrected negative sign, of determined risk measure, and $T$ is the out-of-sample number of observations. $\text{Cost}$, $\text{Cost}^b$ and $\text{Cost}^c$ refer to metrics for model selection obtained from the loss functions used to compute, respectively, $R_{\rho}$, $R_{\rho}^{b}$ and $R_{\rho}^{c}$. 

As observed in Table \ref{tab:res},  $R_\rho$ based on more aggressive risk measures, i.e., those that assume larges values such as ES and ML, tend to display higher values.
Our risk measures are well below VaR and ES, exhibit a smoother behavior that is indicated by standard deviations and ranges, beyond much smaller values for cost criterion. We can explain this because our approach considers, beyond past observations of the financial position, costs from underestimation and overestimation, while VaR and ES only deal with the historic position. With respect to $R_{\rho}^b$, their mean values do not change when using different risk measures and coincide with Maximum Loss, where we quantify risk by the value of worst case. For example, in whole sample, mean value of  $ML =  R^b_{EL} = R^b_{MSD} =  R^b_{EVaR} = R^b_{ML} = 3.80$.

It is also verified that the cost criterion $R_\rho^b$, computed through the three metrics, generally coincide with the values estimated for  $ML$. 
The pattern remains in the sub-samples. We can justify this by extreme values observed in their probability distribution, as descriptive statistics and graphical illustration of Quantile$_2$. 
We also observed that cost criterion, computed  by Cost$^b$ and Cost$^c$,  can take on negative values, i.e., do not fulfill Non-negativity axiom,  and their value, even for the series of returns, is not zero.
We justify this because $R_{\rho}^{b}$ and $R_{\rho}^{c}$ 
consider $E[-xG]$ such as the regret cost. 
For our risk measures, this does not happen because we penalize only the regret costs pertaining to amount where the realized outcome is higher than capital requirement, rather than its total value.

 Regarding our sub-samples, we maintain the results, but some discrepancies arise. We observe in the first sub-sample that log-returns have a positive tendency with high volatility. However, risk measures exhibit smaller and less volatile values, in relation to the  whole sample. Aggregated costs are higher compared to second sub-sample. This pattern links to the yield rates, with larger values for this period and stable evolution. In the second sub-sample, we inverse all the patterns. Again, the yield rates seem to be the major determining factor. 
 
Concerning the performance of the measures, in relation to Cost and  Cost$^c$, with the exception of EL and MSD computed under dual set, our risk measures have the best results. 
As for Cost$^b$, we obtain the lower cost criteria by $R_{\rho}^b$, as expected, and we achieve the highest cost criteria, in general, by  $R_{\rho}^c$. 
In summary, these findings corroborate the results obtained when using the whole sample, evidencing that our risk measures lead to more parsimonious and less costly capital determinations concerning risk measures recommended by regulatory agencies and measures proposed for a purpose similar to ours. This is explained because they consider not only the losses that are taking into account but also the gain opportunities.

	\section{Conclusion}\label{conclusion}

	In this paper, we propose a risk measurement approach that optimally balances costs from capital determination overestimation and underestimation. The objective is to obtain a real value that minimizes the expected value of the sum between both costs. We adopt a robust framework, where we consider the supremum for such minimizer and minimum based on a set of probability measures. We relate our approach with the dual representation of coherent risk measures since we based it on a supremum of expectations over probabilities. We perform an adaptation of the dual representation of usual discrete probability spaces. We develop some theoretical results that guarantee the solution for this problem, develop properties that our risk measure fulfills, and characterize the resulting minimum cost as a deviation measure.

 In an empirical example, we estimate the capital determination using our approach, compared to risk measures that also focus on minimizing costs mentioned and to usual capital requirement determinations implied by Basel and Solvency accords, VaR$^{0.01}$ and ES$^{0.025}$. Results indicate our approach leads to less costly and more parsimonious charges. Our risk measures reflect the temporal evolution of the data, indicating their practical utility. 


Our results are valid to risk management in other fields, such as reliability, ambient environment, and health, for instance. In these areas, as many others where risk management is a growing concern and research field, to have a measurement procedure that balances costs from excessive protection and lack of safeguards is necessary and desired. It is valid to think of extrapolating such a robust optimization problem over a set of probabilities to other contexts, such as portfolio strategies and outside finance. We also suggest conducting generalizations of our proposed approach to dynamic and multivariate frameworks.

	\bibliographystyle{elsarticle-harv}
	\bibliography{refTese}
	
	\singlespace
	\newpage

	\begin{figure}[H]
		\centering
		\includegraphics[width=1\linewidth]{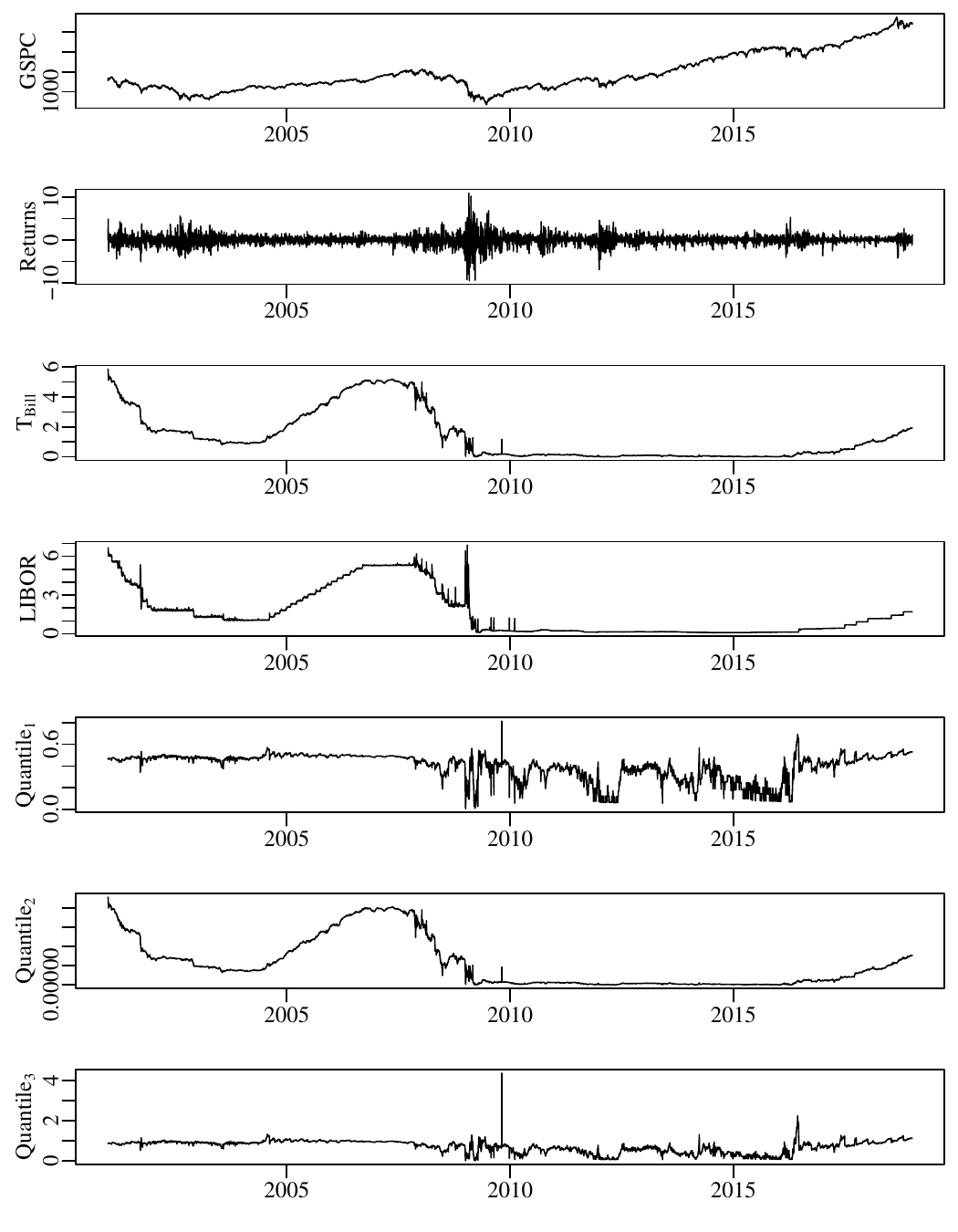}
		\caption{Daily observations from January 2001 to May 2018 of the S\&P 500 adjusted closing price and log-returns ($X$ in \%), yearly three months maturity U.S. Treasury Bill yield ($G$ in \%), yearly U.S. Dollar based Overnight London Interbank Offered Rate ($L$ in \%), and the probability for quantile of $X$ represented by $\frac{G}{G+L}$ ($R_{EL}$), $G$ ($R_{EL}^b$), and $\frac{G}{L}$ ($R_{EL}^c$) -- solutions in the fashion of Proposition \ref{prp:rob}, respectively.}
		\label{fig:plot_datarm}
	
	\end{figure}

	\begin{figure}[H]
		\centering
		\includegraphics[width=1\linewidth]{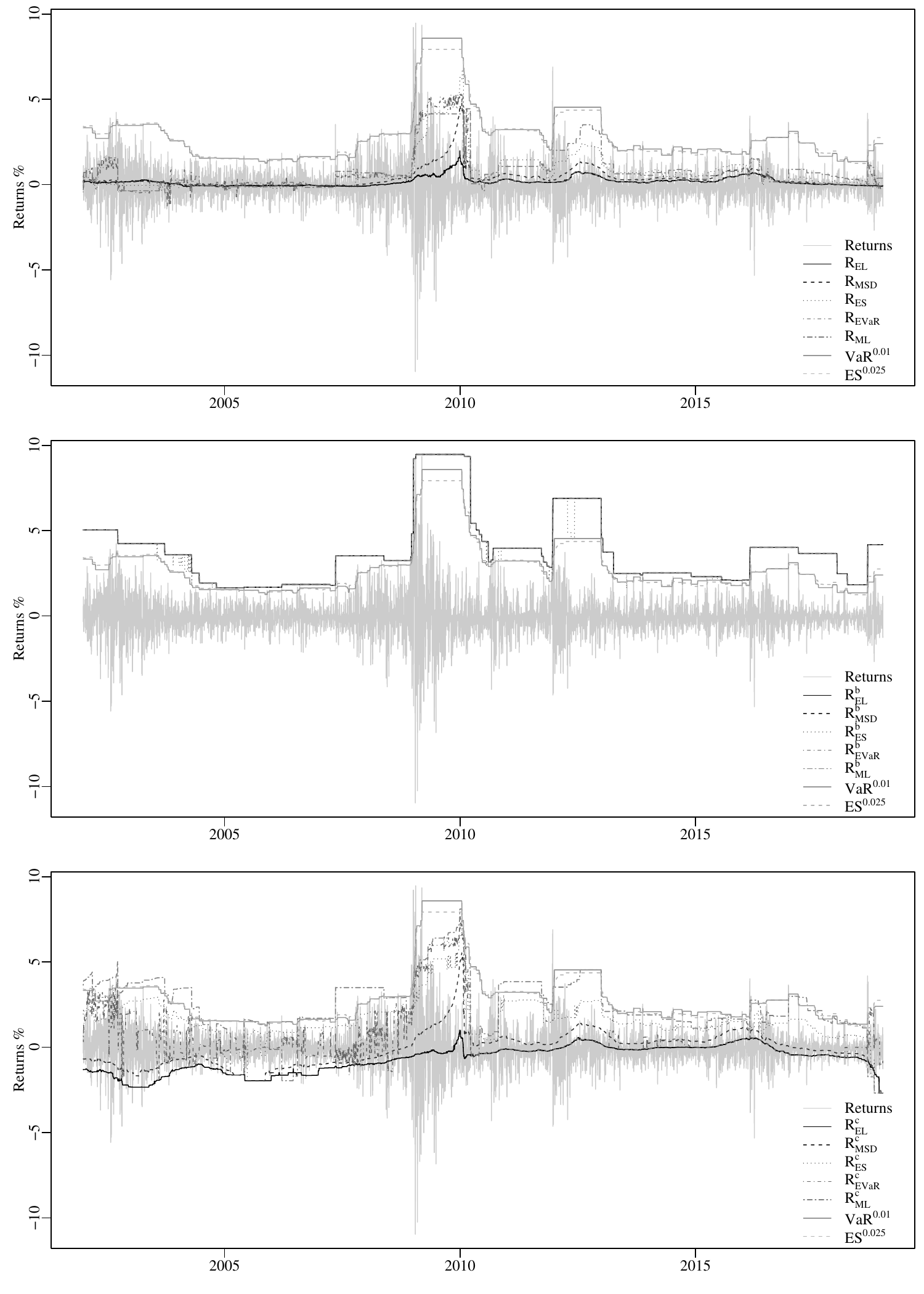}
		\caption{Time series with daily frequency of the log-returns (with adjusted sign) of the S\&P 500 and estimated $R_\rho$, $R_{\rho}^b$, $R_{\rho}^c$, $VaR^{0.01}$ and $ES^{0.025}$ for whole sample, which comprehends January 2001 to May 2018.}
		\label{fig:plot_res}
			
	\end{figure}

\begin{table}[H]
\tiny
		\centering
		\caption{Mean, standard deviation (Dev.), minimum (Min.), maximum (Max.) and cost criterion for daily log-returns of the S\&P 500 and estimated $R_\rho$, $R_{\rho}^b$, $R_{\rho}^c$, $VaR^{0.01}$ and $ES^{0.025}$ for the whole sample (2001 to 2018) and the sub-samples (2001 to  2008 and 2009 to 2018). }
				\begin{center}
		\begin{tabular}{lrrrrrrr}
			\hline
			Whole sample (2001 - 2018) & Mean & Dev. & Min. & Max. & Cost & Cost$^b$ & Cost$^c$\\ 
			\hline
Returns & -0.02 & 1.19 & -10.96 & 9.47 & 0 & $<$ -0.01 & $<$ -0.01 \\
  $R_{EL}$ & 0.17 & 0.26 & -0.12 & 1.99 & 0.16 & 1311.67 & 0.08 \\ 
  $R_{MSD}$ & 0.37 & 0.57 & -0.09 & 4.77 & 0.16 & 1113.30 & 0.07 \\ 
  $R_{ES}$ & 0.70 & 1.08 & -0.32 & 6.74 & 0.16 & 942.94 & 0.10 \\ 
  $R_{EVaR}$ &  0.51 & 1.12 & -1.27 & 5.36 & 0.17 & 1064.63 & 0.11 \\  
  $R_{ML}$ & 0.85 & 1.14 & -0.95 & 4.45 & 0.18 & 823.51 & 0.12 \\
  $R_{EL}^b$ &3.80 & 2.05 & 1.65 & 9.47 & 0.59 & 18.01 & 0.58 \\ 
  $R_{MSD}^b$ & 3.80 & 2.05 & 1.65 & 9.47 & 0.59 & 18.01 & 0.58 \\
  $R_{ES}^b$ & 3.77 & 2.03 & 1.45 & 9.47 & 0.59 & 18.57 & 0.58 \\ 
  $R_{EVaR}^b$ &  3.80 & 2.05 & 1.65 & 9.47 & 0.59 & 18.01 & 0.58 \\
  $R_{ML}^b$ &3.80 & 2.05 & 1.65 & 9.47 & 0.59 & 18.01 & 0.58 \\ 
 $R_{EL}^c$ & -0.69 & 0.77 & -2.68 & 1.01 & 0.33 & 3779.33 & 0.04 \\
 $R_{MSD}^c$  & -0.12 & 1.02 & -2.68 & 5.54 & 0.27 & 2594.78 & 0.05 \\ 
  $R_{ES}^c$ & 1.58 & 1.28 & -2.68 & 6.82 & 0.31 & 592.72 & 0.27 \\ 
  $R_{EVaR}^c$ & 0.59 & 1.59 & -1.93 & 7.70 & 0.21 & 1304.19 & 0.14 \\ 
  $R_{ML}^c$ & 2.43 & 1.80 & -2.68 & 8.11 & 0.49 & 608.32 & 0.43 \\
$VaR^{0.01}$ & 2.86 & 1.67 & 1.26 & 8.58 & 0.44 & 58.24 & 0.43 \\ 
$ES^{0.025}$  & 2.83 & 1.56 & 1.24 & 7.93 & 0.44 & 52.94 & 0.44 \\ 
$ML$ &  3.80 & 2.05 & 1.65 & 9.47 & 0.59 & 18.01 & 0.58 \\ 
   Quantile$_1$ & 0.39 & 0.12 & 0.01 & 0.81 & 0.19 & 1004.44 & 0.15 \\ 
   Quantile$_2$ & $<$ 0.01 & $<$0.01 & $<$0.01 & $<$0.01 & 0.16 & 1559.62 & 0.09 \\ 
   Quantile$_3$ & 0.71 & 0.30 & 0.01 & 4.36 & 0.25 & 759.13 & 0.22 \\   
\hline	
First sub-sample (2001 - 2008)	& Mean & Dev. & Min. & Max. & Cost & Cost$^b$ & Cost$^c$\\ 
\hline	
Returns & 0.02 & 1.35 & -10.96 & 9.47 & 0 & $<$-0.01 & $<$-0.01 \\ 
$R_{EL}$ & 0.03 & 0.14 & -0.12 & 0.61 & 0.14 & 751.40 & 0.07 \\ 
 $R_{MSD}$ & 0.10 & 0.21 & -0.08 & 1.19 & 0.14 & 690.98 & 0.08 \\ 
$R_{ES}$  &0.14 & 0.49 & -0.32 & 2.88 & 0.14 & 683.88 & 0.08  \\ 
$R_{EVaR}$ & 0.30 & 0.73 & -1.27 & 4.41 & 0.15 & 603.69 & 0.10 \\ 
$R_{ML}$  & 0.33 & 0.81 & -0.95 & 4.15 & 0.15 & 592.20 & 0.10 \\ 
$R_{EL}^b$ & 3.26 & 1.65 & 1.65 & 9.47 & 0.49 & 8.65 & 0.49  \\ 
$R_{MSD}^b$ & 3.26 & 1.65 & 1.65 & 9.47 & 0.49 & 8.65 & 0.49 \\ 
$R_{ES}^b$ & 3.24 & 1.65 & 1.45 & 9.47 & 0.49 & 8.65 & 0.49 \\ 
$R_{EVaR}^b$ & 3.26 & 1.65 & 1.65 & 9.47 & 0.49 & 8.65 & 0.49 \\ 
$R_{ML}^b$ & 3.26 & 1.65 & 1.65 & 9.47 & 0.49 & 8.65 & 0.49 \\ 
$R_{EL}^c$ & -1.41 & 0.48 & -2.33 & -0.26 & 0.30 & 2708.23 & 0.04 \\ 
 $R_{MSD}^c$  & -0.93 & 0.55 & -1.95 & 1.01 & 0.24 & 1995.36 & 0.04 \\ 
$R_{ES}^c$ &1.43 & 1.09 & -1.95 & 3.66 & 0.27 & 319.48 & 0.24\\ 
 $R_{EVaR}^c$  & 0.62 & 1.30 & -1.93 & 4.99 & 0.18 & 571.04 & 0.12 \\ 
$R_{ML}^c$  & 2.35 & 1.76 & -1.95 & 5.13 & 0.43 & 345.76 & 0.39 \\ 
$VaR^{0.01}$ &2.46 & 1.21 & 1.26 & 8.58 & 0.37 & 31.27 & 0.37 \\ 
$ES^{0.025}$  & 2.49 & 1.19 & 1.31 & 7.93 & 0.37 & 27.87 & 0.37\\ 
$ML$ & 3.26 & 1.65 & 1.65 & 9.47 & 0.49 & 8.65 & 0.49 \\ 
 Quantile$_1$ & 0.46 & 0.07 & 0.01 & 0.57 & 0.17 & 485.33 & 0.13 \\ 
  Quantile$_2$ & $<$0.01 & $<$0.01 & $<$0.01 & $<$0.01 & 0.14 & 779.46 & 0.08 \\ 
  Quantile$_3$ & 0.89 & 0.17 & 0.01 & 1.33 & 0.22 & 333.13 & 0.20 \\ 
			\hline
Second sub-sample (2009 - 2018) 	& Mean & Dev. & Min. & Max. & Cost & Cost$^b$ & Cost$^c$\\ 
		\hline
		Returns & -0.04 & 0.94 & -5.32 & 6.90 & 0 & $<$-0.01 & $<$-0.01 \\ 
$R_{EL}$  & 0.22 & 0.20 & -0.09 & 0.77 & 0.01 & 467.51 & 0.01 \\ 
$R_{MSD}$ & 0.41 & 0.33 & -0.09 & 1.33 & 0.02 & 373.59 & 0.01 \\ 
$R_{ES}$ & 0.80 & 0.59 & -0.14 & 2.39 & 0.02 & 251.79 & 0.01 \\ 
 $R_{EVaR}$  & 0.24 & 0.22 & -0.24 & 1.37 & 0.02 & 456.50 & 0.01 \\ 
$R_{ML}$ & 0.93 & 0.76 & -0.35 & 3.51 & 0.02 & 223.24 & 0.02 \\ 
 $R_{EL}^b$  & 3.62 & 1.45 & 1.83 & 6.90 & 0.08 & 9.35 & 0.08 \\ 
$R_{MSD}^b$ & 3.62 & 1.45 & 1.83 & 6.90 & 0.08 & 9.35 & 0.08 \\ 
  $R_{ES}^b$L & 3.57 & 1.40 & 1.61 & 6.90 & 0.08 & 9.90 & 0.08 \\ 
$R_{EVaR}^b$ & 3.62 & 1.45 & 1.83 & 6.90 & 0.08 & 9.35 & 0.08 \\ 
$R_{ML}^b$& 3.62 & 1.45 & 1.83 & 6.90 & 0.08 & 9.35 & 0.08 \\ 
$R_{EL}^c$  & -0.15 & 0.45 & -2.68 & 0.58 & 0.02 & 899.79 & $<$0.01 \\ 
$R_{MSD}^c$ & 0.30 & 0.59 & -2.68 & 1.41 & 0.02 & 546.55 & 0.01 \\ 
$R_{ES}^c$ & 1.36 & 0.93 & -2.68 & 3.09 & 0.03 & 269.65 & 0.02 \\ 
$R_{EVaR}^c$ & 0.01 & 0.42 & -1.04 & 2.31 & 0.02 & 725.03 & 0.01 \\ 
$R_{ML}^c$ & 2.10 & 1.37 & -2.68 & 4.55 & 0.05 & 256.51 & 0.03 \\ 
$VaR^{0.01}$ & 2.61 & 0.92 & 1.36 & 4.72 & 0.05 & 26.96 & 0.05 \\ 
$ES^{0.025}$ & 2.57 & 0.88 & 1.24 & 4.59 & 0.05 & 25.05 & 0.05 \\ 
$ML$ &  3.62 & 1.45 & 1.83 & 6.90 & 0.08 & 9.35 & 0.08 \\ 
  Quantile$_1$ & 0.34 & 0.13 & 0.06 & 0.69 & 0.02 & 410.79 & 0.02 \\ 
  Quantile$_2$ & $<$0.01 & $<$0.01 & $<$0.01 & $<$0.01 & 0.02 & 636.13 & 0.01 \\ 
  Quantile$_3$ & 0.57 & 0.30 & 0.06 & 2.24 & 0.03 & 337.17 & 0.02 \\ 
			\hline
		\end{tabular}
		\label{tab:res}
				\end{center}
				
	\end{table}

\end{document}